%% file: main.tex
\documentclass[11pt]{article}
\usepackage[margin=1in]{geometry}
\input{preamble}
\usepackage{thmtools}
\usepackage{mathtools}

\title{Rise and Shine Efficiently! Tight Bounds for Adversarial Wake-up} %
\date{}

\author{Peter Robinson\\
\small{School of Computer \& Cyber Sciences}\\
\small{Augusta University}
\and
Ming Ming Tan \\
\small{School of Computer \& Cyber Sciences}\\
\small{Augusta University}
}

\newcommand{\kt}{\ensuremath{\mathsf{KT}}}
\newcommand{\good}{\ensuremath{\mathit{Sml}}}
\newcommand{\bad}{\ensuremath{\mathit{Bad}}}

\newcommand{\corr}{\ensuremath{\mathit{Cor}}}
\newcommand{\id}{\ensuremath{\text{id}}}
\newcommand{\nih}{\ensuremath{\mathsf{NIH}}}
\newcommand{\supp}{\ensuremath{\mathsf{supp}}}
\newcommand{\arad}{\ensuremath{\rho_{\textsl{\tiny awk}}}} 
\newcommand{\fastwakeup}{\ensuremath{\mathsf{FastWakeUp}}}
\newcommand{\pfw}{\ensuremath{\mathsf{MultiPhaseWakeUp}}}
\newcommand{\cen}{\mathsf{CEN}}
\newcommand{\advice}{\mathsf{Adv}}

\newcommand{\fc}{\mathsf{fc}}
\newcommand{\nxt}{\mathsf{next}}
\newcommand{\port}{\mathsf{port}^{-1}}
\newcommand{\stdport}{\mathsf{port}}

\newcommand{\msg}[1]{\langle #1 \rangle}

\begin{document}
\maketitle
\thispagestyle{empty}
\pagestyle{empty}
\begin{abstract}
We study the wake-up problem in distributed networks, where an adversary awakens a subset of nodes at arbitrary times, and the goal is to wake up all other nodes as quickly as possible by sending only few messages.
We prove the following lower bounds:
\begin{itemize}
\item 
We first consider the setting where each node receives advice from an oracle who can observe the entire network, but does not know which nodes are awake initially. 
More specifically, we consider the $\kt_0$ $\local$ model with advice, where the nodes have no prior knowledge of their neighbors. 
We prove that any randomized algorithm must send $\Omega\lt( \frac{n^{2}}{2^{\beta}\log n} \rt)$ messages if nodes receive only $O(\beta)$ bits of advice on average.
In contrast to prior lower bounds under the $\kt_0$ assumption, our proof does \emph{not} rely on an edge-crossing argument, but uses information-theoretic techniques instead.

\item 
For the $\kt_1$ assumption, where each node knows its neighbors' IDs from the start, we show that any $(k+1)$-time algorithm requires $\Omega\lt( n^{1+1/k} \rt)$ messages.
Our result is the first super-linear (in $n$) lower bound, for a problem that does not require individual nodes to learn a large amount of information about the network topology, which may be of independent interest.
This applies even to the synchronous $\kt_1$ $\local$ model, where the computation is structured into rounds and messages can be of unbounded length.
\end{itemize}
To complement our lower bound results, we present several new algorithms:
\begin{itemize} 
\item 
We give an asynchronous $\kt_1$ $\local$ algorithm that solves the wake-up problem with a time and message complexity of $O\lt( n\log n \rt)$ with high probability.
\item 
We introduce the notion of \emph{awake distance} $\arad$, which is upper-bounded by the network diameter, and present a synchronous $\kt_1$ $\local$ algorithm that takes $O\lt( \arad \rt)$ rounds and sends $O\lt( n^{3/2}\sqrt{\log n} \rt)$ messages with high probability. 
We also extend these ideas to obtain a time complexity of $O\lt( \arad \log^3n \rt)$ rounds as well as a message complexity of $O\lt( n \log^3n \rt)$ messages, both of which are optimal up to log-factors.
\item We give deterministic advising schemes in the asynchronous $\kt_0$ $\congest$ model (with advice).
In particular, we obtain an $O\lt( \arad\log^2n \rt)$-time advising scheme that sends $O\lt( n\log^2n \rt)$ messages, while requiring $O\lt( \log^2n \rt)$ bits of advice per node.
These bounds are optimal in all three complexity measures up to polylogarithmic factors. 
\end{itemize}
\end{abstract}
\newpage %
\setcounter{tocdepth}{2} %
\tableofcontents
\newpage
\pagestyle{plain}
\setcounter{page}{1}
\input{intro}

\input{lb}

\input{kt1_algorithm}

\input{up}

\input{spanner}

\input{conclusion}
\appendix
\section*{Appendix}
\input{information}

\input{basic}
\input{goodkt0}
\input{lb_algo}
\input{app_kt1_algo}
\input{app_kt1_sync}
\bibliographystyle{alpha}
\bibliography{refs}
\end{document}

%% file: preamble.tex
\usepackage{amsmath,amsbsy}
\usepackage{amsthm} %
\usepackage{microtype} %
\usepackage{lmodern} %
\usepackage{setspace}
\usepackage{color}
\usepackage{xcolor}
\usepackage{algorithm}
\usepackage[noend]{algpseudocode}
\usepackage{multirow}
\usepackage{xspace}
\usepackage{threeparttable}
\usepackage{booktabs}
\definecolor{MyGreen}{rgb}{0.1333,0.5451,0.1333}
\usepackage[linktocpage=true, 
	pagebackref=true,colorlinks,
  linkcolor=black,
	citecolor=MyGreen,
	bookmarks,bookmarksopen,bookmarksnumbered]
	{hyperref}

\let\originalleft\left
\let\originalright\right
\renewcommand{\left}{\mathopen{}\mathclose\bgroup\originalleft}
\renewcommand{\right}{\aftergroup\egroup\originalright}

\renewcommand{\Pr}{\operatorname*{\textbf{\textup{Pr}}}}
\DeclareMathOperator*{\II}{\textbf{\textup{I}}}

\DeclareMathOperator*{\HH}{\textbf{\textup{H}}}

\DeclareMathOperator*{\EE}{\textbf{\textup{E}}}

\newcommand{\set}[1]{\{ #1 \}}
\newcommand{\Set}[1]{\left\{ #1 \right\}}

\newcommand{\lb}{\left}
\newcommand{\rb}{\right}
\newcommand{\lt}{\left}
\newcommand{\rt}{\right}
\newcommand{\md}{\middle}

\DeclareMathOperator{\poly}{\ensuremath{\mathrm{poly}}}

\newcommand{\local}{\ensuremath{\mathsf{LOCAL}}}
\newcommand{\LOCAL}{\local}
\newcommand{\congest}{\ensuremath{\mathsf{CONGEST}}}

\makeatletter
\newtheorem*{rep@theorem}{\rep@title}
\newcommand{\newreptheorem}[2]{%
\newenvironment{rep#1}[1]{%
\def\rep@title{#2 \ref{##1}}%
\begin{rep@theorem}[restated]}%
{\end{rep@theorem}}}
\makeatother
\newreptheorem{lemma}{Lemma}
\newreptheorem{theorem}{Theorem}
\newreptheorem{corollary}{Corollary}

\newboolean{short}
\newcommand{\onlyShort}[1]{\ifthenelse{\boolean{short}}{#1}{}}
\newcommand{\onlyLong}[1]{\ifthenelse{\boolean{short}}{}{#1}}

\newtheorem{problem}{Open Problem}
\theoremstyle{plain}
\newtheorem{lemma}{Lemma}
\newtheorem{theorem}{Theorem}
\newtheorem{corollary}{Corollary}
\newtheorem{claim}{Claim} %
\newtheorem{fact}{Fact} %
\newtheorem{open problem}{Open Problem}

\usepackage[most]{tcolorbox}
\tcbuselibrary{theorems}
\tcolorboxenvironment{definition}{lower separated=false,
                colback=white!97!blue,
colframe=white, fonttitle=\bfseries,
colbacktitle=white!60!gray,
coltitle=black,
enhanced,
boxed title style={colframe=black},
}
\tcolorboxenvironment{theorem}{lower separated=false,
                colback=white!97!blue,
colframe=white, 
fonttitle=\bfseries,
colbacktitle=white!60!gray,
coltitle=black,
enhanced,
boxed title style={colframe=black},
}
\tcolorboxenvironment{corollary}{lower separated=false,
                colback=white!97!blue,
colframe=white,
fonttitle=\bfseries,
colbacktitle=white!60!gray,
coltitle=black,
enhanced,
boxed title style={colframe=black},
}
\tcolorboxenvironment{open problem}{colback=white!95!gray,
colframe=white, 
fonttitle=\bfseries,
colbacktitle=white!60!gray,
coltitle=black,
enhanced,
}

\newcommand{\ann}[1]{%
\text{\footnotesize(#1)}\quad}

%% file: intro.tex
\section{Introduction} \label{sec:intro}

In this paper, we study distributed algorithms in networks where every node is either awake or asleep, and our goal is to understand whether it is possible to wake up all sleeping nodes efficiently. 
Our investigation is motivated by well-established networking standards such as Wake-on-LAN and Wake-on-Wireless-LAN~\cite{amd_20213,wikipedia_wake-on-lan}, where a sleeping node only listens to wake-up messages (called ``magic packets'') on its network adapter and does not perform any other computation.
Allowing unused nodes to be in a sleep state can potentially lead to significantly reduced energy consumption in large networks and may increase the performance-per-watt ratio of data centers; e.g., see ~\cite{gandhi2012sleep}.

To formalize the setting, we consider an asynchronous communication network represented by a graph of $n$ nodes and $m$ edges.
Each node is running an instance of a distributed algorithm by exchanging messages over its incident communication links (i.e., edges) of the network.
Initially, an arbitrary non-empty subset of the nodes are awake, whereas all other nodes continue to sleep until they are woken up, e.g., by receiving a message from some already awake node. 
Two key metrics of a distributed algorithm are its {time complexity} and its {message complexity}. 
The former quantifies the worst case number of time units until the algorithm has reached its goal, whereas the latter counts the total number of messages sent throughout the execution. 
In this work, we investigate the time and message complexity of the \emph{wake-up problem}, which was formally defined in \cite{fraigniaud2006oracle}. 
That is, an adversary chooses a network topology and decides which nodes it wakes up and at what time, while the algorithm's goal is to wake up all other nodes as quickly and by sending as few messages as possible. 

While it is clear that, without further assumptions, any algorithm must take time steps that is at least as large as the network diameter $D$ for solving the wake-up problem, determining the achievable message complexity is less obvious, as prior work in this area suggests that the initial knowledge of the nodes can affect the achievable bounds significantly.
Two well-studied assumptions are known as $\kt_1$~\cite{AGPV88} and $\kt_0$~\cite{peleg}: 
The $\kt_1$ assumption stipulates that each node starts out knowing who else it is connected to, which is an adequate abstraction for modern-day IP networks.
On the other hand, \emph{$\kt_0$} (also known as the \emph{port numbering model}) requires that each node sends messages to its neighbors by using integer port numbers that are unrelated to the neighbors' IDs.
Given that nodes have significantly less knowledge under the $\kt_0$ assumption, it is not too surprising that sending a number of messages that is proportional to the number of edges in the network is a requirement for solving virtually any problems of interest (see \cite{jacm15}), and this easily extends to the wake-up problem.
One possible way to circumvent this lower bound is to consider the $\kt_1$ assumption, which allows breaking the $\kt_0$ message complexity barrier for problems such as constructing a minimum spanning tree (MST)\onlyShort{.}\onlyLong{, as we elaborate in more detail in Section~\ref{sec:related}.}
However, understanding the time and message complexity of the fundamental wake-up problem is still unresolved in  networks under the $\kt_1$ assumption. 
Our work takes a first step towards closing these gaps in literature.

\onlyLong{
Apart from the knowledge of nodes about their immediate neighbors, we also consider the impact of providing each node with more general knowledge about the network on the complexity bounds, i.e., we explore the \emph{information sensitivity}~\cite{fraigniaud2009distributed} of the wake-up problem, which captures how the difficulty of a problem scales with the amount of knowledge provided to the nodes. 
To this end, we consider the $\kt_0$ setting where each node is equipped with ``advice''~\cite{fraigniaud2006oracle}, which is a bit string of a certain length that is computed as a function of the entire network topology.
Understanding the impact of advice on the complexity bounds may also provide some insights into the amount of ``global communication'' that is needed for solving this problem, when considering the hybrid model of communication introduced by Augustine, Hinnenthal, Kuhn, Scheideler, and Schneider in \cite{DBLP:conf/soda/AugustineHKSS20}, where they show that global communication allows for a significant speedup for the shortest path problem compared to the standard assumption of local-only communication.
                  
Regarding the impact of advice on the wake-up problem, we point out that Fraigniaud, Ilcinkas, and Pelc~\cite{fraigniaud2006oracle} give a simple and elegant way of solving the wake-up problem with $O\lt( n \rt)$ messages and a total advice length of $\Theta\lt( n\log n \rt)$ bits (summed up over all nodes), however, two fundamental questions were left unanswered by their work:
(1) Can we leverage advice to reduce \emph{both}, messages and time?
(2) Is it possible to obtain a logarithmic upper bound on the \emph{maximum} advice length assigned to any node?
We address both of these questions by giving new advising schemes that yield near optimal time and message complexity, thus taking a step towards fully understanding the impact of advice on the performance of distributed algorithms in asynchronous networks.
}

\onlyLong{
\subsection{Computing Model and Problem Definition} \label{sec:prelim}
We consider a communication network abstracted as an undirected unweighted graph $G$ of $n$ nodes and $m$ edges, where each node of $G$ is associated with a computer executing a distributed algorithm, and each edge corresponds to a bidirectional communication channel. 
Every node $u$ has a unique integer ID, denoted by $\id(u)$, which is chosen from a range of size polynomial in $n$, and we assume that nodes have knowledge of a constant factor upper bound on $\log n$ \footnote{All logarithms are assumed to be of base $e$ unless stated otherwise.}. 

A node $u$ may communicate with its neighbors in $G$ via message passing, which allows $u$ to send (possibly distinct) messages over any subset of its incident edges whenever taking a computing step.
Unless stated otherwise, we assume an \emph{asynchronous network} and hence any message sent may be subject to unpredictable but finite delay.
However, all communication channels are error-free and deliver messages in a FIFO manner.
We study the wake-up problem in both the  $\congest$ and $\local$ model, see \cite{peleg}. 
In the $\congest$ model, each message has a size of at most $O(\log n)$ bits, while in the $\local$ model, the focus is purely on the impact of locality and the messages length is unbounded.

Each node performs local computing steps, which may be triggered by the receipt of a message. 
We follow the standard convention in distributed computing that local computation happens instantaneously and is ``free'', in the sense that we do not restrict the amount of computation a node may perform in a single step.  
As we are interested in randomized algorithms, each node has access to a private source of unbiased random bits that it may consult as part of its local computation. 

Initially, each node is either \emph{asleep} or \emph{awake}. 
Upon receiving a message, a sleeping node permanently changes its status to awake and starts executing the algorithm. 
In particular, any message sent to a currently-sleeping node $u$ is not lost, but processed by $u$ upon awakening. 

\paragraph{Initial Knowledge of the Nodes: $\kt_0$ vs.\ $\kt_1$.}
When considering the \emph{$\kt_0$ assumption}, a node with degree $\deg(u)$ has integer ports labeled $1,\ldots,\deg(u)$ and the algorithm needs to specify the port when sending a message.
We formalise the port mapping at a node $v$ as follows. 
Let $N_v$ denote the set of neighbors of $v$ and $\deg(v)$ be the degree of $v$. 
The port mapping at $v$ is a bijective function $\stdport_v: [\deg(v)] \mapsto N_v$. 
That is, port $i$ at $v$ will lead to its neighbor $\stdport(i)$, and we define the $\port_v$ as the inverse of $\stdport_v$.
That is, $\port_v(u)$ is the port number that $v$ uses to identify the incident edge that leads to its neighbor $u$.%
For example, in Figure~\ref{fig:lb_kt0} on page~\pageref{fig:lb_kt0}, node $v_i$ is connected via port $3$ to node $u_1$, which in turn is connected to $v_i$ by port $1$.}
We emphasize that for $\kt_0$, $v$ has no prior knowledge of the concrete mapping and that we assume that the other endpoint of an edge also learns the port connection, if one of its two vertices sends a message over its respective port.

The \emph{$\kt_1$ assumption} provides more knowledge to the algorithm than $\kt_0$, since each node starts out knowing the IDs of all its neighbors, which it can use when sending messages.
It is well known, see e.g., \cite{DBLP:conf/podc/KingKT15,podc15} that $\kt_1$ is powerful enough for implementing graph sketching techniques~\cite{AGM-soda12}.

\paragraph{Adversary.} An {adversary} determines the network topology, the node IDs, and the set of initially awake nodes. 
In the case of $\kt_0$, the adversary also determines each individual node's port mapping. 
Apart from controlling the message delays, it may also decide to wake up a currently-sleeping node at any point in the execution.\footnote{This is a crucial difference to related work in the context of energy and awake complexity (e.g., see \cite{DBLP:conf/podc/ChangDHHLP18,DBLP:conf/podc/ChatterjeeGP20}), where the assumption is that the algorithm (and not the adversary) controls the wake up schedule.}
When considering randomized algorithms, we assume that the adversary is \emph{oblivious}, in the sense that it must decide the delay of in-transit messages and which nodes to wake up (and at what time) without knowing the state of the nodes, which includes their private random bits.

\onlyLong{
\paragraph{Advising Scheme.} \onlyLong{In Sections~\ref{sec:lb_kt0} and \ref{sec:advice},}\onlyShort{In Section~\ref{sec:advice},} we augment the $\kt_0$ assumption by considering the ``computing with advice'' framework~\cite{fraigniaud2006oracle,fraigniaud2009distributed}, which was recently studied in the context of distributed graph algorithms in \cite{10.1145/3662158.3662805}. 
An \emph{advising scheme} consists of (1) an \emph{oracle} that takes as input the network $G$ and equips each node with a bit string, called \emph{advice}, and (2) a distributed algorithm $\mathcal{A}$ that may leverage the advice to solve the problem at hand. 
When talking about the time or message complexity of an advising scheme, we are referring to the performance of algorithm $\mathcal{A}$.
Intuitively speaking, we can think of an \emph{oracle}, who observes the entire network before the start of the execution, and uses this information when computing an advice string for each node. 
Unless stated otherwise, we assume that the oracle does not know the set of initially-awake nodes. 

\paragraph{Problem Definition.}
We say that an algorithm (or an advising scheme) solves the \emph{wake-up problem} with probability $1-\epsilon$ if, for any network and any execution as determined by the adversary, the probability that every node wakes up eventually is at least $1-\epsilon$.
}
\onlyLong{
\subsection{Complexity Measures and Awake Distance} \label{sec:awake_distance}

\paragraph{Time Complexity.}
To normalize the time taken by an algorithm, we use $\tau$ to denote an upper bound on the message delay of any message in the execution.
We use $\tau$ to define the length of one \emph{time unit} and emphasize that $\tau$ is not part of the knowledge of the nodes in the asynchronous model.
The \emph{time complexity} of an algorithm is the worst case number of time units from the time at which the first node wakes up until the latest point in time at which any message is received. 
}

\paragraph{Awake Distance.}
In many real-world applications, it may be desirable to ensure that nodes wake up sooner than time proportional to the network diameter, if there are awake nodes located closer to them.
This motivates introducing a more fine-grained way of quantifying the performance of an algorithm with respect to time.
Given a graph $G$ and a set of initially-awake nodes $A_0$ that are being awoken by the adversary, we define the \emph{awake distance} 
\begin{align}
\arad = \arad(G,A_0) = \max_{u \in G} \text{dist}_G(A_0,u), \label{eq:arad}
\end{align}
where $\text{dist}_G(A_0,u)$ is the shortest hop distance of $u$ to some node in $A_0$.
Note that $\arad$ is equivalent to the time complexity of the (message-inefficient) standard flooding algorithm.
As elaborated in more detail in Section~\ref{sec:contributions} and summarized in Table~\ref{fig:table}, we have designed wake-up algorithms that send few messages and achieve a time complexity proportional to $\arad$.

\onlyLong{
\paragraph{Message Complexity.}
Another important complexity measure of a distributed algorithm is its \emph{message complexity}, which is the worst case number of messages sent over all nodes in any execution.
We also consider the \emph{expected message complexity}, where the expectation is computed over the private randomness of the nodes with respect to the worst case execution produced by the adversary. 

\paragraph{Advice Length.}
In the context of equipping nodes with advice, we aim to bound the \emph{average length} as well as the \emph{maximum length of advice} assigned to any node; analogously to the time and message complexity, we compute these values over the worst case execution as determined by the adversary.

\onlyLong{
\subsection{Additional Related Work} \label{sec:related}
Several works on leader election and minimum spanning tree computation (see, e.g.,
\cite{gallager1983distributed,afek1991time,singh1992leader,kutten2020singularly,dufoulon2022almost,DBLP:conf/podc/Kutten0T023})
study the impact of adversarially awoken nodes.  All of these works consider the
$\kt_0$ assumption, where the known lower bound of $\Omega\lt( m \rt)$~\cite{jacm15}
poses an insurmountable barrier on the message complexity for solving any
problem that is at least as hard as single-source broadcast, and thus it is unsurprising that this extends to the wake-up problem in both asynchronous and
synchronous networks.  
We point out that \cite{dufoulon2022almost} also give an algorithm in the asynchronous model for constructing an MST under the $\kt_1$ assumption using $o(m)$ messages. However, their result relies on the algorithm of King and Mashregi~\cite{mashreghi2019brief,DBLP:journals/dc/MashreghiK21}, who provide an elegant algorithm for computing an MST
that achieves a message complexity of $O(n^{3/2}\log n)$ in the asynchronous
$\kt_1$ $\congest$ model, assuming that all nodes are awake initially. 
We explain why this approach does not work under adversarial wake-up:
Initially, in their algorithm, each node becomes a ``star'' with probability $1/\sqrt{n \log n}$, whereby non-stars with a degree of greater than $\sqrt{n}\log^{3/2}n$ (``high-degree nodes''), remain silent until they receive a message.
If the adversary simply wakes up exactly a single high-degree node, it will become a non-star that remains silent forever with probability $1 - 1/\sqrt{n\log n}$, thus resulting in a high probability of error.
}
\onlyShort{
We discuss additional related work such as gossip algorithms in the attached full version of the paper.
}
}

Another class of algorithms that are relevant for obtaining low message complexity in synchronous networks are gossip protocols~\cite{karp2000randomized}, since they ensure that each node initiates communication with only a single peer in each step.
In more detail, a node who already knows the message may \emph{push} information, while a currently-uninformed node may \emph{pull} information from one of its neighbors.
Consequently, a gossip protocol that terminates in $T$ rounds will have a message complexity of $n \cdot T$.
For instance, the seminal works of \cite{DBLP:conf/stoc/Censor-HillelHKM12,haeupler2015simple} show how to solve broadcast with polylogarithmic overhead in the synchronous gossip model under the $\kt_1$ assumption.
Unfortunately, it is unclear how to directly use gossip protocols for the wake-up problem, as these results crucially rely on the availability of both, push and pull for information exchange, whereas in our setting, asleep nodes are prohibited from using pull.
Even though it is known that gossip using only push-based information exchange solves broadcast efficiently in regular graphs with good expansion properties~\cite{sauerwald2011rumor}, this does not extend to general graphs.\footnote{Consider a complete graph $H$ and a single vertex $v$ connected to $H$ by one edge. The resulting graph has constant vertex expansion. However, if the rumor is initially at some node in $H$, it takes $\Omega\lt( n \rt)$ time steps in expectation until $v$ receives the rumor via push-based broadcast.}

\subsection{Contributions and Technical Challenges} \label{sec:contributions} 

\input{contributions}

%% file: contributions.tex
We present several novel algorithms and lower bounds that demonstrate the interplay between time, messages, and length of advice. 
Table~\ref{fig:table} summarizes our results.

\onlyLong{
\begin{table*}[t] \label{table: results}
\centering
\small
\begin{threeparttable}
\begin{tabular*}{1.01\textwidth}{l c c c l c}
\toprule
& Time & Messages & Advice & Model & Random \\
\midrule
\multicolumn{2}{l}{\textbf{Algorithms}:}\\
Theorem~\ref{thm:dfs_many} & $O(n\log n)^{*}$ & $O(n\log n)^{*}$ & - & \footnotesize{async. \textcolor{blue}{$\kt_1$} $\local$} & yes\\ 
Theorem~\ref{thm:kt1_fw} & $O(\arad)$ & $O(n^{3/2}\sqrt{\log n})^{*}$ & - & \footnotesize{\phantom{a}\textsl{sync.} \textcolor{blue}{$\kt_1$} $\local$} & yes\\ 
Theorem~\ref{thm:pfw} & $O(\arad\log^3n)$ & $O(n{\log^3 n})^{*}$ & - & \footnotesize{\phantom{a}\textsl{sync.} \textcolor{blue}{$\kt_1$} $\local$} & yes\\ 
\cite{fraigniaud2006oracle}, Cor.~\ref{cor:basic} & $O\lt( D \rt)$ & $O\lt( n \rt)$ & $O\lt( n \rt)$%
 & \footnotesize{async. \textcolor{brown}{$\kt_0$} $\congest$} & no\\
Theorem~\ref{thm:scheme_broadcast}(A) & $O(D)$ & $O(n^{3/2})$ & $O(\sqrt{n}\log n)$ & \footnotesize{async. \textcolor{brown}{$\kt_0$} $\congest$} & no\\ 
Theorem~\ref{thm:scheme_cen}(B) & $O\lt( D\log n \rt)$ & $O\lt( n \rt)$ & $O\lt( \log n \rt)$ & \footnotesize{async. \textcolor{brown}{$\kt_0$} $\congest$} & no\\
Theorem~\ref{thm:advice_spanner} & $O\lt( k\,\arad\log n \rt)$ & $O\lt( k\,n^{1+1/k} \rt)$ & $O\lt( n^{1/k}\log^2 n \rt)$ & \footnotesize{async. \textcolor{brown}{$\kt_0$} $\congest$} &no\\
Corollary~\ref{cor:advice_spanner} & $O\lt( \arad\log^2 n \rt)$ & $O\lt( n \log^2 n \rt)$ & $O\lt( \log^2 n \rt)$ & \footnotesize{async. \textcolor{brown}{$\kt_0$} $\congest$} &no\\
\midrule
\multicolumn{2}{l}{\textbf{Lower Bounds}:}\\
 Theorem~\ref{thm:lb_kt0}$^{\$}$ & - & $\le\frac{n^{2}}{2^{\beta+4}\log_2n}$ & $\Omega\lt( \beta \rt)$ & \footnotesize{\textsl{sync.} \textcolor{brown}{$\kt_0$} $\local$} & yes\\ 
Theorem~\ref{thm:lb}$^{\dagger}$ & $\le k+1$ & $\Omega\lt( n^{1+1/k} \rt)$ & - & \footnotesize{\textsl{sync.} \textcolor{blue}{$\kt_1$} $\local$} & yes\\
\bottomrule
\end{tabular*}
{\footnotesize
  \begin{tablenotes}
  \item[$*$] With high probability.
  \item[$\$$] The bound on the advice length holds, if the expected message complexity is at most $\frac{n^{2}}{2^{\beta+4}\log_2n}$.
  \item[$\dagger$] Any algorithm must satisfy the message complexity bound if it terminates within $k+1$ time units.\\
  \end{tablenotes}
}
\caption{\small Algorithms and Lower Bounds for the Wake-up Problem. Column ``Advice'' refers to the maximum length of advice per node, unless stated otherwise. Column ``Random'' indicates whether nodes have access to random bits. Parameter $\arad$ refers to the awake distance, see Section~\ref{sec:awake_distance}.} 
\label{fig:table}
\end{threeparttable}
\end{table*}
}
\onlyShort{
\begin{center}
\begin{table*}[t] \label{table: results}
\centering
\begin{threeparttable}
\begin{tabular*}{0.7\textwidth}{l c c c l c}
\toprule
& Time & Messages & Advice & Model & Random \\
\midrule
\multicolumn{2}{l}{\textbf{Algorithms}}\\
 & $O(n\log n)^{*}$ & $O(n\log n)^{*}$ & - & \footnotesize{async. \textcolor{blue}{$\kt_1$} $\local$} & yes\\ 
 & $O(\arad)$ & $O(n^{3/2}\sqrt{\log n})^{*}$ & - & \footnotesize{\phantom{a}\textsl{sync.} \textcolor{blue}{$\kt_1$} $\local$} & yes\\ 
 & $O\lt(D\rt)$ & $O\lt(n^{3/2}\rt)$ & $O\lt(\sqrt{n}\log n\rt)$ & \footnotesize{async. \textcolor{brown}{$\kt_0$} $\congest$} & no\\ 
 & $O\lt( D\log n \rt)$ & $O\lt( n \rt)$ & $O\lt( \log n \rt)$ & \footnotesize{async. \textcolor{brown}{$\kt_0$} $\congest$} & no\\
 & $O\lt( k\,\arad\log n \rt)$ & $O\lt( k\,n^{1+1/k} \rt)$ & $O\lt( n^{1/k}\log^2 n \rt)$ & \footnotesize{async. \textcolor{brown}{$\kt_0$} $\congest$} &no\\
 & $O\lt( \arad\log^2 n \rt)$ & $O\lt( n \log^2 n \rt)$ & $O\lt( \log^2 n \rt)$ & \footnotesize{async. \textcolor{brown}{$\kt_0$} $\congest$} &no\\
\midrule
\multicolumn{2}{l}{\textbf{Lower Bounds}}\\
 & - & $\le\frac{n^{2}}{2^{\beta+4}\log_2n}$ & $\Omega\lt( \beta \rt)$ & \footnotesize{\textsl{sync.} \textcolor{brown}{$\kt_0$} $\local$} & yes\\ 
 & $\le k+1$ & $\Omega\lt( n^{1+1/k} \rt)$ & - & \footnotesize{\textsl{sync.} \textcolor{blue}{$\kt_1$} $\local$} & yes\\
\bottomrule
\end{tabular*}
{\footnotesize
  \begin{tablenotes}
  \item[$*$] With high probability.\\
  \end{tablenotes}
}
\caption{\small Algorithms and Lower Bounds for the Wake-up Problem. Column ``Advice'' refers to the maximum length of advice per node, unless stated otherwise. Column ``Random'' indicates whether nodes have access to random bits. Parameter $\arad$ refers to the awake distance.} 
\label{fig:table}
\end{threeparttable}
\end{table*}
\end{center}
}

\onlyLong{
\subsubsection{A Lower Bound on the Advice for Randomized Algorithms in $\kt_0$ (Section~\ref{sec:lb_kt0}, page~\pageref{sec:lb_kt0})}
}
\onlyShort{
\subsubsection{A Lower Bound on the Advice for Randomized Algorithms in $\kt_0$}
}
The work of \cite{fraigniaud2006oracle} gives an elegant combinatorial argument that the total advice assigned to the nodes must be $\Omega\lt( n\log n \rt)$ bits, for obtaining a message complexity of $O\lt( n \rt)$ when assuming $\kt_0$.
However, their argument only holds for deterministic algorithms and does not reveal the actual message complexity required, e.g., when restricting the advice to $o(\log n)$ bits per node, i.e., $o(n\log n)$ bits in total.

We present a new result under the $\kt_0$ assumption that not only works for randomized algorithms with advice, but also results in a polynomial improvement on the lower bound on the message complexity, when the advice provided to each node is small.

\newcommand{\thmLbKTZero}{%
Let $\mathcal{A}$ be a randomized advising scheme that solves the wake-up problem in the (synchronous or asynchronous) $\kt_0$ $\local$ model and errs with probability $\epsilon<\tfrac{1}{2\log_2 n}$.
For any positive $\beta \leq \log_2n$, if the expected message complexity of $\mathcal{A}$ is at most $ \frac{n^{2}}{2^{\beta+4}\log_2 n}$, then the average length of advice per node is at least $\frac{1}{6} \cdot (\beta-2-o(1)) = \Omega\lt( \beta \rt)$ bits. 
In particular, an advice length of $o\lt( \log n \rt)$ bits per node requires an expected message complexity of $\Omega\lt( n^{2-\alpha} \rt)$, for any constant $\alpha>0$.
This holds even if the oracle knows the set of awake nodes and even if we assume shared randomness. 
}
\begin{theorem} \label{thm:lb_kt0}
\thmLbKTZero
\end{theorem}

For proving Theorem~\ref{thm:lb_kt0}, we define a lower bound graph $\mathcal{G}$, where every node $v_i$ in a large set of the nodes (called center nodes) has exactly one edge to a sleeping neighbor $w_i$, who cannot be woken up by anyone else.
Considering that nodes do not know their port mappings, it is not too difficult to show that the center nodes would essentially need to send messages across most of their ports, if we assumed the standard $\kt_0$ $\local$ model without advice. 
The main technical challenge in proving Theorem~\ref{thm:lb_kt0} emanates from the fact that the oracle gets to see all port mappings when computing the advice. 
For instance, consider the special case where the message complexity is $O(n^{2 - \alpha})$, for some small constant $\alpha>0$. 
We need to take into account the possibility that the oracle encodes the $O\lt( {\log n} \rt)$ bits required for representing the port number leading to $v_i$'s sleeping neighbor $w_i$. %
For instance, the oracle could partition the port number for $w_i$ into $\omega(1)$ pieces and store each piece among a subset of the neighbors of $v_i$.
This would, in fact, suffice for $v_i$ identifying (and waking up) $w_i$, as it can receive messages from $\omega\lt( \log n \rt)$ of its neighbors, each of arbitrary size, without violating the message complexity bound.
To avoid this pitfall, our proof leverages that there are many such nodes that each need to find their sleeping neighbor and, consequently, the oracle cannot successfully encode \emph{all} of these ports among the nodes in the network without using at least $\Omega\lt( n\log n \rt)$ bits in total.

\onlyLong{
\subsubsection{A Lower Bound on the Message Complexity in $\kt_1$ (Section~\ref{sec:lb_kt1}, page~\pageref{sec:lb_kt1})}
}
\onlyShort{
\subsubsection{A Lower Bound on the Message Complexity in $\kt_1$}
}

We also consider the challenging $\kt_1$ setting (without advice), where nodes start out knowing their neighbors' IDs.
In other words, the knowledge of the IDs associated with an edge is shared between its two endpoints, and this enables the use of powerful graph sketching techniques~\cite{AGM-soda12,KKM-soda13,DBLP:conf/podc/KingKT15} that allow discovering ``outgoing'' edges (such as the edge $\set{v_i,w_i}$) with a polylogarithmic overhead of messages, as elaborated in Section~\ref{sec:contributions}. 
Furthermore, as shown in \cite{soda21}, graph sketching allows solving \emph{any} graph problem with just $\tilde O(n)$ message complexity in the synchronous $\kt_1$ $\congest$ model, albeit at the expense of a prohibitively large time complexity, which stands in stark contrast to the unconditional lower bound of $\Omega\lt( m \rt)$ known to hold for the $\kt_0$ assumption (without advice). %
Thus, we focus on \emph{time-restricted} algorithms and present the first trade-off on the time and achievable message complexity for the wake-up problem under the $\kt_1$ assumption.
\newcommand{\thmLb}{%
Consider any integer $k \in [3,o(\log n)]$, and let $\mathcal{A}$ be a randomized Las Vegas algorithm that solves the wake-up problem in the (synchronous or asynchronous) $\kt_1$ $\local$ model.
If $\mathcal{A}$ takes at most $k+1$ units of time in every execution with awake distance $\arad=1$, then the expected message complexity of $\mathcal{A}$ is at least $\Omega\lt( n^{1+{1}/{k}} \rt)$.
}
\begin{theorem} \label{thm:lb}
\thmLb
\end{theorem}

When disregarding the message complexity, a straightforward flooding algorithm solves the awake problem in optimal $\arad$ time.
Thus, an immediate corollary of Theorem~\ref{thm:lb} is that optimality cannot be achieved in both, time and messages, under the $\kt_1$ assumption, even if we allow unbounded messages and assume synchronous rounds.

We point out that, prior to this work, almost no lower bounds on the message complexity were known for problems in the general $\kt_1$ setting.
With the exception of \cite{soda21}, all previous lower bounds (e.g., \cite{AGPV88,pai2021can,DBLP:conf/innovations/DufoulonPPP024,kutten2024tight}) use edge-crossing arguments that only hold for comparison-based algorithms, which restrict nodes to behave the same when observing ``order-equivalent'' IDs among their neighbors.
While the lower bound for graph spanners in \cite{soda21} does hold for general algorithms under the $\kt_1$ assumption, it exploits the fact that many nodes needs to identify a large set of incident edges to correctly compute a spanner, and it is unclear how to extend this approach to problems where the output size at individual nodes is negligible, as is the case for the wake-up problem.
To the best of our knowledge, Theorem~\ref{thm:lb} is the first super-linear lower bound on the message complexity that holds for general $\kt_1$ algorithms for a problem that does not require most nodes to learn a large amount of information about the graph topology.
Moreover, our lower bound holds even in the $\local$ model, whereas all of the above-mentioned works use a bottleneck argument that crucially exploits the $\congest$ model.

For proving Theorem~\ref{thm:lb}, we modify the lower bound construction $\mathcal{G}$ used for $\kt_0$ to define a new class of graphs $\mathcal{G}_k$ with large girth, while retaining the property that every node $v_i$ in a certain set has exactly one \emph{crucial neighbor} $w_i$ among its $\Theta(n^{1/k})$ neighbors, who cannot be woken up by any other node. Our overall goal is to argue that identifying this one sleeping neighbor is hard, similarly as for $\kt_0$. However, there are several technical challenges that we need to overcome: 
\begin{itemize} 
\item 
From the perspective of a node $v_i \in V$, it takes only $\sigma = O\lt(\frac{\log n}{k}\rt)$ bits of information for determining which of its $\Theta(n^{1/k})$ incident edges leads to its crucial neighbor $w_i$.
Thus, at first glance, it might seem that the graph sketching techniques introduced by Ahn, Guha, and McGregor~\cite{AGM-soda12}, which is known to yield $O(n\poly\log n)$ message complexity for spanning tree construction (see \cite{DBLP:conf/podc/KingKT15}), could be leveraged to wake-up every node with only a polylogarithmic message overhead.

\item Note that we want to show that the result holds even in the $\kt_1$ $\local$ model, where a node $v_i \in V$ may receive a polynomial number of bits during the execution. In particular, this means we cannot utilize a ``bottleneck argument'' based on small cuts to argue that $v_i$ does not learn enough information about $w_i$ quickly enough, as is commonly done in the $\congest$ model (see, e.g., \cite{DBLP:journals/siamcomp/SarmaHKKNPPW12}).
\end{itemize}
\onlyLong{
For ruling out a graph sketching approach and making our proof work in the $\LOCAL$ model, we need to carefully argue that most of the information received by $v$ is not useful in determining the ID of its sleeping neighbor.
We point out that our lower bound proof does \emph{not} involve information-theoretic techniques, but instead uses a combinatorial argument to show a limit on the information flow in high girth graphs, which rests on the assumption that the algorithm is time-restricted.
Note that the limitation to time-restricted algorithms is not a mere artifact of our proof, as we show in Theorem~\ref{thm:dfs_many} that it is possible to solve the problem in $\tilde O(n)$ time and $\tilde O(n)$ messages in the $\kt_1$ $\local$ model. 
}

\onlyLong{
\subsubsection{Near Optimal Message Complexity in the Asynchronous $\kt_1$ $\local$ Model (Section~\ref{sec:kt1_algo_async}, page~\pageref{sec:kt1_algo_async})}
}
\onlyShort{
\subsubsection{Upper Bounds}
We defer the discussion of our algorithms to the full paper~\cite{wakeup_arxiv} .
 }
\onlyShort{
\subsubsection{Near Optimal Message Complexity in the Asynchronous $\kt_1$ $\local$ Model}
}

We design several novel algorithms for the adversarial wake-up problem.
We first discuss the asynchronous $\kt_1$ $\local$ model (without advice). 
Note that an event happens \emph{with high probability} (w.h.p) if it occurs with probability at least $1-\frac{1}{n^{\Omega(1)}}$.
\newcommand{\thmDFSMany}{
  There exists a randomized Las Vegas algorithm for the wake-up problem in the asynchronous $\kt_1$ $\local$ model that, with high probability, has a time and message complexity of $O(n\log n)$.
}
\begin{theorem} \label{thm:dfs_many}
  \thmDFSMany
\end{theorem}

The main idea behind the algorithm in Theorem~\ref{thm:dfs_many} is to combine depth-first search (DFS) explorations starting from each node awoken by the adversary with randomly chosen ranks, with the goal that (1) there is a limit on the number of DFS traversal that may visit the same nodes, and (2) the adversary is unlikely to delay the termination time by strategically waking up nodes after a certain time.
While random ranks have been used for various problems such as leader election in the synchronous setting (e.g., see \cite{DBLP:journals/dc/KhanKMPT12}), we need a different analysis for  proving the time and message complexity bounds that takes into account the impact of adversarial wake-up and the asynchronous scheduling.

We remark that the message complexity in Theorem~\ref{thm:dfs_many} is almost optimal, since $\Omega\lt( n \rt)$ is a trivial lower bound on the message complexity of the wake-up problem. 
In other words, this also illustrates the necessity of restricting the time complexity when showing the super-linear message complexity lower bound in Theorem~\ref{thm:lb}.

\onlyLong{
\subsubsection{Achieving Wake-Up in $\tilde O(\arad)$ Synchronous Rounds in $\kt_1$ $\local$ (Section~\ref{sec:kt1_algo_sync}, page~\pageref{sec:kt1_algo_sync})}
}
\onlyShort{
\subsubsection{Achieving Wake-Up in $\tilde O(\arad)$ Synchronous Rounds in $\kt_1$ $\local$}
}

We present a simple message-efficient algorithm that achieves a time complexity of $O\lt( \arad \rt)$, in the synchronous model, where the computation proceeds in lock-step rounds.\onlyLong{\footnote{We do not, however, assume that nodes have a common clock, i.e., when a node wakes up, it does not know how many rounds have passed since the adversary woke up the first node.}} %

\newcommand{\thmFW}{
  There is a randomized Las Vegas algorithm that solves the wake-up problem in the synchronous $\kt_1$ $\local$ model in $10\arad$ rounds and sends $O\lt( n^{3/2}\sqrt{\log n}  \rt)$ messages with high probability.
}
\begin{theorem} \label{thm:kt1_fw}
  \thmFW
\end{theorem}
To give some intuition for our algorithm $\fastwakeup$ we focus on the special case of $\arad=1$, i.e., the initially-awake nodes form a dominating set.
The trivial algorithm, where every one of these nodes simply broadcasts, clearly wakes up every node in just a single round, albeit at a high message complexity cost, when many nodes are awake.
The main idea behind our algorithm is to reduce this cost by subsampling a small set of the awake nodes that each build a constant-depth BFS tree in a message-efficient manner, which slightly increases the time complexity to $10$ rounds.
We show that this prevents too many close-by awake nodes from attempting to wake up the same set of sleeping nodes.
While its message complexity leaves a polynomial gap to the optimal bound of $O(n)$, instantiating Theorem~\ref{thm:lb} tells us that any constant-time algorithm must send at least $n^{1+\Omega( 1 )}$ messages.

Then, we combine the ideas of Algorithm~$\fastwakeup$ with the cover construction algorithm of Elkin~\cite{elkin2006faster} to obtain an algorithm that is time- as well as message-optimal (up to log-factors).

\newcommand{\thmPFW}{
Algorithm~$\pfw$ is a randomized Las Vegas that wakes up all nodes in the synchronous $\kt_1$ $\local$ model in $O\lt( \arad \log^3n \rt)$ rounds and has a message complexity of $O\lt( n\log^3n \rt)$ with high probability.
}
\begin{theorem} \label{thm:pfw}
\thmPFW
\end{theorem}

\onlyLong{
\subsubsection{Asynchronous Algorithms with Advice in $\kt_0$ $\congest$ (Section~\ref{sec:advice}, page~\pageref{sec:advice})}
}
\onlyShort{
\subsubsection{Asynchronous Algorithms with Advice in $\kt_0$ $\congest$}
}

As a warm-up, we start our investigation in this setting with a straightforward generalization of the work of \cite{fraigniaud2006oracle}, who give an advising scheme that essentially concatenates every incident tree edge of a spanning tree as advice.
This scheme attains optimal message complexity, albeit at the cost of a worst case advice length of $O(n\log n)$ bits per node.
First, it is easy to see that their approach takes $O\lt( D \rt)$ time, where $D$ is the diameter of the graph, when instructing the oracle to use a BFS tree instead of an arbitrary spanning tree.%
\onlyLong{
Moreover, in Appendix~\ref{app:basic}, we give a simple argument that shows that we can shave off a log-factor of the maximum advice length:
}%

\begin{corollary}[follows from \cite{fraigniaud2006oracle}]\label{cor:basic}
There is an advising scheme in the asynchronous $\kt_0$ $\congest$ model that solves the wake-up problem in $O(D)$ time, $O(n)$ messages, a maximum advice length of $O\lt( n \rt)$ per node, and average advice length of $O(\log n)$ per node.
\end{corollary}

Next, we present two novel advising schemes that significantly reduce the maximum advice length per node. 
The first one yields optimal time, whereas the second one matches the best possible bound on the message complexity.
\newcommand{\thmSchemeBasic}{%
For each of the following cases, there is a deterministic advising scheme in the asynchronous $\kt_0$ $\congest$ model that satisfies the stated complexity bounds:
\begin{enumerate}
 \item[(A)]  $O(D)$ time, $O(n^{3/2})$ messages, a maximum advice length of $O(\sqrt{n}\log n)$ and an average advice length of $O(\log n)$;
 \item[(B)]  $O(D\log n)$ time, $O(n)$ messages, and a maximum advice length of $O(\log n)$ bits. 
\end{enumerate}
}
\begin{theorem} \label{thm:scheme_basic} \label{thm:scheme_broadcast}\label{thm:scheme_cen}
  \thmSchemeBasic
\end{theorem}

For Part~(A), we show how to reduce the advice length by (almost) a $\sqrt{n}$-factor compared to Corollary~\ref{cor:basic}, albeit at the cost of increasing the message complexity. 
Then, in Part~(B), we develop a technique to distribute the information necessary for recovering the incident BFS tree edges of a node among its neighbors, which we call \emph{child-encoding scheme}.
This allows us to obtain optimal messages and nearly-optimal time with just $O(\log n)$ bits of advice per node.

While the diameter presents an insurmountable lower bound on the time complexity when considering all possible sets of awake nodes, the awake distance (defined in Sec.~\ref{sec:prelim}) gives a more fine-grained classification of the time in terms of the distance that nodes have to the set of awake nodes.%
\onlyLong{
In Section~\ref{sec:spanner}, we present a tradeoff between time (in terms of the awake distance $\arad$), messages, and length of advice: 
}
\onlyShort{
In the full paper, we present a tradeoff between time (in terms of the awake distance $\arad$), messages, and length of advice: 
}

\newcommand{\thmAdviceSpanner}{
There exists a deterministic advising scheme with a maximum advice length of $O(n^{1/k}\log^2 n)$ bits that allows solving the wake-up problem in time $O\lt(k\cdot\arad\log n\rt)$ with a message complexity of $O\lt( k\,n^{1+1/k}\log n  \rt)$ in the asynchronous $\kt_0$ $\congest$ model.
}
\begin{theorem} \label{thm:advice_spanner}
  \thmAdviceSpanner
\end{theorem}

We obtain Theorem~\ref{thm:advice_spanner} by using the aforementioned child-encoding scheme to encode the edges of a suitable graph spanner as advice. 
An immediate implication of this result is that we can match the optimal bounds on messages and time (in terms of the awake distance) up to polylogarithmic factors.  

\newcommand{\corAdviceSpanner}{
There is a deterministic advising scheme for the wake-up problem with a maximum advice length of $O(\log^2 n)$ bits in $O(\arad\log^2 n)$ time using $O\lt( n\log^2 n \rt)$ messages in the asynchronous $\kt_0$ $\congest$ model.
}
\begin{corollary} \label{cor:advice_spanner}
  \corAdviceSpanner
\end{corollary}

Note that a consequence of Theorem~\ref{thm:lb_kt0} is that the length of advice in Corollary~\ref{cor:advice_spanner} is at most a logarithmic factor larger than the minimum possible amount.

%% file: lb.tex
\section{Lower Bounds} \label{sec:lb}

We will prove our lower bounds by assuming the synchronous $\LOCAL$ model, where messages may be of arbitrary size and where the computation is structured into \emph{synchronous rounds}~\cite{peleg}. %
We first introduce a lower bound graph construction for the $\kt_0$ assumption that we will use for proving Theorem~\ref{thm:lb_kt0} in Section~\ref{sec:lb_kt0}.
Subsequently, in Section~\ref{sec:lb_kt1}, we describe the necessary changes to make this class of graphs suitable for proving Theorem~\ref{thm:lb}, i.e., the $\kt_1$ lower bound.

\paragraph{The Class of Lower Bound Graphs $\mathcal{G}$.}
We consider a graph of $3n$ nodes, whose vertex set is the union of three sets $U=\set{u_1,\dots,u_n}$, $V=\set{v_1,\dots,v_n}$, and $W=\set{w_1,\dots,w_n}$, and we call the nodes in $V$ \emph{center nodes}.
The edges  include a perfect matching $\set{v_1,w_1},\dots,\set{v_n,w_n}$ between $V$ and $W$, and these are the only edges incident to nodes in $W$.
We say that $w_i$ is a \emph{crucial neighbor} of $v_i$.
The edges between $U$ and $V$ correspond to a complete bipartite graph, which guarantees that all nodes in $V$ have a degree of $n+1$.
We assign the node IDs according to some arbitrary fixed permutation of the integers in $[3n]$.\footnote{We follow the convention that $[n]$ denotes the set of integers $\set{1,\dots,n}$.}
Initially, we assume that all center nodes (i.e., nodes in $V$) are awake, whereas all other nodes are asleep. 

To sample a concrete graph $G \sim \mathcal{G}$,  we choose the port assignments uniformly at random from all valid port mappings.
We emphasize that this guarantees that the port mappings of the nodes are mutually independent.

\paragraph{The Needles-in-Haystack ($\nih$) Problem.}
Our lower bound proofs in Sections~\ref{sec:lb_kt0} and \ref{sec:lb_kt1} rest on leveraging the difficulty in discovering the edge between each awake center node $v_i \in V$ and its corresponding crucial neighbor $w_i \in W$, who can only be woken up by a direct message from $v_i$. 
To make this more concrete, we say that an algorithm solves the \emph{needles-in-haystack} ($\nih$) problem if, for each node $v_i$, it holds that
\begin{itemize}  
\item $v_i$ outputs the port connecting to $w_i$, when considering $\kt_0$, and
\item $v_i$ outputs the ID of node $w_i$ when assuming $\kt_1$.
\end{itemize}

We now show that an algorithm for the wake-up problem gives rise to an algorithm for the $\nih$ problem at negligible additional cost:

\begin{lemma} \label{lem:lb_simul}
Let $\mathcal{A}$ be an algorithm that solves the wake-up problem with probability at least $1-\epsilon$, in $T$ time, with a worst-case advice length $\lambda$ per node, and by using $M$ messages in expectation.
Then, there exists an algorithm $\mathcal{B}$ that solves the $\nih$ problem with probability at least $1-\epsilon$, with a time complexity of $T+1$, an expected message complexity of at most $M+n$, and an advice length of $\lambda$ bits.
This holds under the $\kt_0$ and also under the $\kt_1$ assumption.
Moreover, $\mathcal{B}$ is deterministic if $\mathcal{A}$ is deterministic.
\end{lemma}

\begin{proof}
We outline how to obtain $\mathcal{B}$ from $\mathcal{A}$.
Since only the nodes in $W$ have degree $1$ and all other nodes have strictly larger degrees, we can instruct each $w_i \in W$ to send a special response message (distinct from all messages produced by $\mathcal{A}$) upon receiving the first message from $v_i$, thus informing $v_i$ that it has successfully done its part. 
This increases the expected message complexity by $n$ and the time complexity by one round.
According to the properties of $\mathcal{A}$, all nodes in $W$ will be woken up with probability at least $1 - \epsilon$ in $T$ rounds.
\end{proof}

Equipped with Lemma~\ref{lem:lb_simul}, we focus on showing a lower bound for an $\nih$ algorithm $\mathcal{B}$, where each node $v_i \in V$ eventually outputs a value defined as follows: 
If $\mathcal{B}$ is correct, then the output of $v_i$ is $\port_{v_i}(w_i)$ for $\kt_0$, and when considering $\kt_1$, $v_i$'s output is $\id(w_i)$; in the case where $\mathcal{B}$ errs, the output of $v_i$ may be arbitrary.
Consequently, for algorithm $\mathcal{B}$, we are not concerned with whether all nodes are woken up successfully, as a node $v_i$ does not necessarily need to send a message across the edge identified by its output for producing a correct output.

\input{lb_kt0}
\input{lb_kt1}

%% file: lb_kt0.tex
\subsection{A Lower Bound on the Message Complexity for $\kt_0$ Algorithms with Advice} \label{sec:lb_kt0}

\begin{reptheorem}{thm:lb_kt0}
\thmLbKTZero
\end{reptheorem}

In the remainder of this section, we prove Theorem~\ref{thm:lb_kt0}. 
Due to Lemma~\ref{lem:lb_simul}, we consider an $\epsilon$-error algorithm $\mathcal{B}$ for NIH problem with an expected message complexity of at most $n^{2}/(2^{\beta+4}\log_2n)+O(n)$, and our goal is to show the claimed lower bound of $\Omega(\beta)$ on the average length of advice.
Our proof relies on several notions from information theory, which we briefly recall in Appendix~\ref{app:tools}.

\paragraph{Random Variables and Notation.}
To precisely capture the amount of knowledge obtained by the nodes over the course of the execution, we need the following definitions. 
We use $X \perp Y$ to say that random variables $X$ and $Y$ are independent. 
For a node $u \in U \cup V \cup W$, we use $Y_{u}$ to denote the advice string of $u$, and recall that $Y_{u}$ is part of $u$'s initial knowledge.
Moreover, we use $\mathbf{Y}$ to denote the concatenation of all the advice strings of the nodes.
We define $\Pi_i$ to be the transcript of the messages received by $v_i$ throughout the execution, excluding any messages sent by $w_i$, and we use $B$ to denote the shared random bit string that is known to all nodes.
Moreover, for a random variable $X$, we use the shorthand $\Pr\lt[ x \rt]$ to denote $\Pr\lt[ X=x \rt]$. 
We use $X_{< i}$ as an abbreviation for $(X_1, X_2, \ldots, X_{i-1})$, and define $X_{\le i}$ analogously. 
We will make use of the event $\corr$, which happens if the algorithm correctly solves the $\nih$ problem, and we also employ the corresponding indicator random variable $\mathbf{1}_{\corr}$.

\paragraph{Event $\good_i$ and Set $S$.}
For each node $v_i \in V$, we define the event $\good_i$ to occur if $v_i$ sends or receives a message over at most $\frac{n}{2^{\beta}}$ of its ports, i.e., the algorithm uses a ``small'' number of $v_i$'s ports. 
Let $S \subseteq V$ be the set of nodes such that, for every $v_i \in S$, we have 
$\Pr\lt[ \good_i \rt]\ge 1 - \frac{1}{2\log_2n}$, where the probability is taken over the random bits of the nodes as well as the random port assignment.

In the next lemma, we show that $S$ is large, i.e., a constant fraction of the nodes in $V$ are likely to send and receive messages over a small number of their ports. 
The proof is based on an averaging argument, which we relegate to Appendix~\ref{app:good_kt0}.
\newcommand{\lemGoodKTZero}{
It holds that $|S| \ge \frac{n}{2}.$
}
\begin{lemma} \label{lem:good_kt0}
  \lemGoodKTZero
\end{lemma}

We are now ready to start deriving a lower bound on the total length of the advice $\mathbf{Y}$ assigned to the nodes. 
Let $X_i$ be the port number at $v_i$ connecting to $w_i$, i.e., $X_i = \port_{v_i}(w_i)$. 
By Lemma~\ref{lem:inf_props}(a) and (b), we have
\begin{align}
|\mathbf{Y}|
&\ge
\HH\lt[ \mathbf{Y} \ \md|\ B \rt] \notag\\ 
&\ge
\II\lt[ X_{\le n} : \mathbf{Y} \ \md|\ B \rt] \notag\\ 
\ann{by Lemma~\ref{lem:inf_props}(c)}
&=
	 \sum_{i=1}^{n} \II\lt[ X_i :  \mathbf{Y} \ \md|\ B, X_{< i} \rt]. \notag
\end{align}
Observing that random variable $X_i$ is independent of $X_{<i}$, as the port mappings are chosen independently for each node, allows us to apply Lemma~\ref{lem:inf_props}(g) to the terms in the sum on the right-hand side, yielding 
\begin{align}
|\mathbf{Y}|
&\ge   
	 \sum_{i=1}^{n} \II\lt[ X_i :  \mathbf{Y} \ \md|\ B \rt] \notag\\
\ann{by \eqref{eq:mutual_cond}}
&= 
	 \sum_{i=1}^{n} 
   \lt( 
   \HH\lt[ X_i \ \md|\ B \rt] 
   -
   \HH\lt[ X_i  \ \md|\ B,\mathbf{Y} \rt] 
   \rt)\notag\\ 
&= 
	 \sum_{i=1}^{n} 
   \lt( 
   \log_2(n+1) 
   -
   \HH\lt[ X_i  \ \md|\ B,\mathbf{Y} \rt] 
   \rt),
\end{align}
where the second last equality follows because $X_i$ is equal to the port connecting $v_i$ to $w_i$, which is independent of the shared randomness $B$, and thus $X_i$ is uniform among the $\deg(v_i)=n+1$ neighbors of $v_i$ (see Lemma~\ref{lem:inf_props}(f)).
Restricting the summation to nodes in set $S$ yields
\begin{align}
\mathbf{Y} \ge 
	 \sum_{v_i\in S}
   \lt( 
   \log_2(n)
   -
   \HH\lt[ X_i  \ \md|\ B,\mathbf{Y} \rt] 
   \rt).
   \label{eq:entropy_term}
\end{align}

The next lemma captures the core technical argument of the $\kt_0$ lower bound, namely, the difficulty of reducing the number of possibilities for $X_i$ without sending many messages, even if nodes are equipped with advice. 

\begin{lemma} \label{lem:lb_contrad}
For all $v_i \in S$, it holds that
$\displaystyle
\HH\lt[ X_i \ \md|\ B, \mathbf{Y} \rt] \le \log_2\lt( \frac{n}{2^{\beta-1}} \rt) + 1 + o(1). 
$
\end{lemma}
\begin{proof}
It holds that
\begin{align}
\HH\lt[ X_i \ \md|\ B, \mathbf{Y} \rt] 
&\le
	\HH\lt[ X_i, \mathbf{1}_{\corr\wedge\good_i} \ \md|\ B, \mathbf{Y} \rt] \notag\\ 
\ann{by Lemma~\ref{lem:inf_props}(e)}
&= 
	\HH\lt[ X_i \ \md|\ B, \mathbf{Y}, \mathbf{1}_{\corr\wedge\good_i} \rt] 
  +
	\HH\lt[ \mathbf{1}_{\corr\wedge\good_i}\ \md|\ B, \mathbf{Y} \rt] 
  \notag\\ 
&\le 
	\HH\lt[ X_i \ \md|\ B, \mathbf{Y}, \mathbf{1}_{\corr\wedge\good_i} \rt] 
  +
	1. 
  \notag\\ 
\end{align}  
Moreover, 
\begin{align}
\HH\lt[ X_i \ \md|\ B, \mathbf{Y}, \mathbf{1}_{\corr\wedge\good_i} \rt]
&\le
 \HH\lt[ X_i \ \md|\ B, \mathbf{Y}, {\corr\wedge\good_i} \rt] 
 +
 \Pr\lt[  \neg(\corr\wedge\good_i) \rt] 
 \HH\lt[ X_i \ \md|\ B, \mathbf{Y}, \neg(\corr\wedge\good_i) \rt] \notag\\ 
&\le
 \HH\lt[ X_i \ \md|\ B, \mathbf{Y}, {\corr\wedge\good_i} \rt] 
 +
 \frac{1}{\log_2n} 
 \HH\lt[ X_i \ \md|\ B, \mathbf{Y}, \neg(\corr\wedge\good_i) \rt],\notag\\ 
\intertext{
since $\Pr\lt[  \neg(\corr\wedge\good_i) \rt] \le \Pr\lt[ \neg\corr \rt] + \Pr\lt[ \neg  \good_i \rt] \le \frac{2}{2\log_2n}$.
  Moreover, plugging the trivial upper bound $\HH\lt[ X_i \ \md|\ B, \mathbf{Y}, \neg(\corr\wedge\good_i) \rt] \le \log_2(n+1) \le \log_2(n) + 1$ into the right-hand side yields
}
\HH\lt[ X_i \ \md|\ B, \mathbf{Y}, \mathbf{1}_{\corr\wedge\good_i} \rt]
&\le
 \HH\lt[ X_i \ \md|\ B, \mathbf{Y}, {\corr\wedge\good_i} \rt] 
 + 1 + \frac{1}{\log_2n}.\notag\\ 
\end{align}  
  
Assume towards a contradiction that 
\begin{align}
\HH\lt[ X_i \ \md|\ B, \mathbf{Y},\corr\wedge\good_i \rt] > \log_2\lt( \frac{n}{2^{\beta-1}} \rt). \label{eq:ent_cont} %
\end{align}
Using the definition of the conditional entropy (see \eqref{eq:conditional_entropy}), we have 
\[
\HH\lt[ X_i \ \md|\ B, \mathbf{Y},\corr\wedge\good_i \rt] \!=\! \sum_{b,y} \Pr\lt[ b,y\ \md|\ \corr\wedge\good_i \rt] \HH\lt[ X_i \ \md|\ b, y,\corr\wedge\good_i \rt],
\] 
and \eqref{eq:ent_cont} implies that
\begin{align}
\sum_{b,y} \Pr\lt[ b,y\ \md|\ \corr\wedge\good_i \rt] \HH\lt[ X_i \ \md|\ b, y,\corr\wedge\good_i \rt] &> \log_2\lt( \frac{n}{2^{\beta-1}} \rt). \label{eq:ent_lb1}
\end{align}
Let $\supp(X)$ denote the size of the support of a random variable $X$, which is the number of values $x$ such that $\Pr\lt[ X \!=\!  x \rt]>0$.
Recalling Lemma~\ref{lem:inf_props}(f), we know that 
$\HH\lt[ X_i \ \md|\ b, y,\corr\wedge\good_i \rt] \le \log_2\lt(\supp(X_i \ \md|\ b, y,\corr\wedge\good_i)\rt)$.
Thus,
\begin{align}
 \EE_{B,\mathbf{Y}}\lt[ \log_2\lt(\supp(X_i \ \md|\ b, y,\corr\wedge\good_i)\rt)  \rt]
 &=
 \sum_{b,y}  \Pr\lt[ b,y\ \md|\ \corr\wedge\good_i \rt]\cdot \log_2\lt(\supp(X_i \ \md|\ b, y,\corr\wedge\good_i)\rt) \notag\\ 
 \ann{by \eqref{eq:ent_lb1}}
 &> \log_2\lt( \frac{n}{2^{\beta-1}} \rt).\notag
\end{align}
It follows that there are events $\set{B\!=\! b,\mathbf{Y}\!=\! y}$ that occur with nonzero probability, conditioned on $\corr \wedge \good_i$, such that $\log_2\lt(\supp(X_i \ \md|\ b, y,\corr\wedge\good_i)\rt) > \log_2 \lt( \frac{n}{2^{\beta-1}} \rt)$. 
Equivalently,
\begin{align}
\supp\lt(X_i \ \md|\ b, y,\corr\wedge\good_i\rt) >  \frac{n}{2^{\beta-1}}, \label{eq:supp_lb}
\end{align}
i.e., there are more than $\frac{n}{2^{\beta-1}}$ (not necessarily uniform) choices for $X_i$ that cannot be ruled out.
Since we condition on $\good_i$, we know that at most $\frac{n}{2^{\beta}}$ of $v_i$'s incident ports are used throughout the execution for sending or receiving messages.
Let $p_{v_i}$ denote the set of ports at $v_i$ that are not used throughout the execution, and observe that the function $\stdport_{v_i}$ restricted to $p_{v_i}$ is completely independent of how the ports are connected at other nodes 
and the shared randomness.
See Figure~\ref{fig:lb_kt0} for a concrete example.
\begin{figure}[t]
  \centering
\includegraphics[scale=1.25]{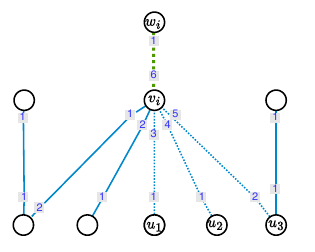}
\caption{
  In this example, we assume that the concatenation of all advice strings given to the nodes, i.e., $\mathbf{Y}$, reveals the port mappings for ports $1$ and $2$ at $v_i$, and we also assume that no messages are sent over its remaining ports, i.e., $p_{v_i}=\set{3,4,5,6}$.
  While $v_i$ may know that the image of the map $\port(p_{v_i})$ is the set $\set{u_1,u_2,u_3,w_i}$, it does not know the individual mappings of the ports in $p_{v_i}$, and thus there are still $(4!)$ possible mappings. (It could be that, conditioned on $\mathbf{Y}$, some of these are more likely than others.)
  Note that neighbor $u_1$ may know that its port $1$ is indeed connected to $v_i$---perhaps, because it has already explored all other ports. 
  However, since there is no dependency between $u_i$'s and $v_i$'s port  numbers, this does not reveal any information about how the individual ports in $p_{v_i}$ are connected to the nodes in $\set{u_1,u_2,u_3,w_i}$. 
}
\label{fig:lb_kt0}
\end{figure}
While keeping \eqref{eq:supp_lb} in mind, we now turn our attention to how a node $v_i \in V$ computes its output for the $\nih$ problem.
Any message received by $v_i$ from nodes in $U$ may depend on $B$, $\mathbf{Y}$, and possibly on how other ports in the network are connected with the notable exception of $p_{v_i}$.
Suppose that $v_i$ receives a message over some previously unused port, and assume that $v_i$ already knows $\mathbf{Y}$. 
Then, it follows that the actual content of the message cannot reveal any new information about $X_i$ (conditioned on $\mathbf{Y}$!), apart from reducing the number of choices for $X_i$ by $1$. 
Note that this argument crucially relies on the $\kt_0$ assumption and the independence of the port assignments for distinct nodes. 
When conditioning on $\corr \wedge \good_i$, it follows that, for every possible transcript of messages $\set{\Pi_i \!=\! \pi_i}$ received by $v_i$, we have
\begin{align}
    \supp\lt(X_i \ \md|\ \pi_i,b, y,\corr\wedge\good_i\rt)
    &\ge 
\supp\lt(X_i \ \md|\ b, y,\corr\wedge\good_i\rt) - \frac{n}{2^{\beta}} \notag\\ 
\ann{by \eqref{eq:supp_lb}}
&> \frac{n}{2^{\beta}},					
  \label{eq:lb_contrad2}
\end{align}
The port $t$ that is output by $v_i$ is a deterministic function given $\pi_i$, $b$, and $v_i$'s advice, which is part of $y$.
Moreover, when conditioning on event $\corr$, we know that $t = X_i$.
Consequently, we have $\supp\lt(X_i \ \md|\ \pi_i,b, y,\corr\wedge\good_i\rt) =1 \leq \frac{n}{2^\beta} $, contradicting \eqref{eq:lb_contrad2}. Note that in the last inequality, we have used the assumption that $\beta \leq \log_2n$, according to the premise of Theorem~\ref{thm:lb_kt0}.
\end{proof}

We now return to \eqref{eq:entropy_term}. 
Applying Lemma~\ref{lem:lb_contrad} and using the bound $|S|\ge \frac{n}{2}$ guaranteed by Lemma~\ref{lem:good_kt0}, we obtain
$|\mathbf{Y}| \ge \frac{n}{2}(\beta-2-o(1)).$
Since the network has $3n$ nodes, we conclude that the average length of advice is at least $\frac{1}{6}\cdot (\beta-2-o(1)) = \Omega(\beta)$ bits. 
This completes the proof of Theorem~\ref{thm:lb_kt0}.

%% file: lb_kt1.tex
\subsection{A Lower Bound on the Message Complexity for Time-Restricted $\kt_1$ Algorithms} \label{sec:lb_kt1}

\paragraph{The Class of Lower Bound Graphs $\mathcal{G}_k$.}
For a given parameter $k$, we define $n = q^k$, where $q$ is a prime power, and $k \ge 3$ be a positive odd integer. 
We modify the graph construction $\mathcal{G}$ defined in Section~\ref{sec:lb} for the $\kt_0$ setting, by keeping the perfect matching edges between $V$ and $W$ as is and, instead of a complete bipartite graph, we add an $n^{\frac{1}{k}}$-regular bipartite graph $H$, which has a girth $\rho$ of at least $k+5$ and consists of $\Omega\lt( n^{1+\frac{1}{k}} \rt)$ edges.
The existence of $H$ with the desired properties follows from the results of Lazebnik and Ustimenko~\cite{girth}.%
\footnote{Strictly speaking, the construction of \cite{girth} results in a disconnected graph consisting of multiple components. 
We could easily make this graph connected by inserting some edges between nodes in $U$.
For the sake of simplicity, we refrain from doing so as our lower bound argument rests on the high girth assumption and does not exploit the connectivity of the graph in any way.}
Analogously to Section~\ref{sec:lb_kt0}, we assume that all center nodes (i.e., nodes in $V$) are awake, whereas all other nodes are asleep. 

We now describe the input distribution from which we sample a concrete instance of the lower bound graph: 
Since we are considering the $\kt_1$ model, each node starts out knowing the IDs of its neighbors. We assume that the IDs of the nodes in $V$ are fixed such that node $v_j$ gets ID $2n+j$, for $j \in [n]$.
For choosing the IDs of the nodes in $U \cup W$, on the other hand, we  uniformly at random choose a permutation from $[2n]$. (Note that this is different from the class of graphs $\mathcal{G}$ used for the $\kt_0$ lower bound, where the IDs were fixed and the port connections were chosen randomly instead.)
Figure~\ref{fig:lb} gives an example of the lower bound graph. %

\begin{figure}[t]
  \centering
\includegraphics[scale=1.35]{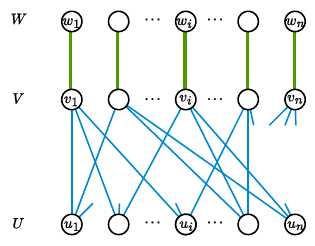}
\caption{\small
The lower bound graph construction for Theorem~\ref{thm:lb}. Initially, all nodes in $U$ and $W$ are asleep, while the center nodes (i.e., nodes in $V$) are awake. Node IDs and some edges are omitted for clarity. To solve the wake-up problem, the center nodes need to send a message over the green thick edges to their respective crucial neighbor. 
}
\label{fig:lb}
\end{figure}

The following properties are immediate from the definition of $\mathcal{G}_k$ and the aforementioned result of \cite{girth}:

\begin{fact} \label{lem:graph_props}
Any graph $G$ sampled from $\mathcal{G}_k$, has the following properties:
\begin{enumerate}
\item All nodes in $V$ have the same degree $d = n^{1/k}+1$.
\item $G$ has $\Omega\lt(n^{1+1/k}\rt)$ edges. 
\item The girth of $G$  is at least $k+5$;
\end{enumerate}
\end{fact}

We remark that $\mathcal{G}_k$ admits a simple $O( \log^2 n )$-time algorithm with a message complexity of $O( n\log^2 n )$, by using the deterministic gossip algorithm of \cite{haeupler2015simple} as a subroutine, which we describe in Appendix~\ref{app:lb_algo} for completeness.
In the remainder of this section, we prove Theorem~\ref{thm:lb}: 

\begin{reptheorem}{thm:lb}
\thmLb
\end{reptheorem}

We focus on deterministic algorithms first; subsequently, we will use Yao's lemma to extend the result to randomized algorithms.
We write $\mu(\mathcal{G}_k)$ to denote the uniform distribution over the graphs in $\mathcal{G}_k$. 
Assume towards a contradiction that there is a deterministic algorithm $\mathcal{A}$
with an expected message complexity of at most $o\lt(n^{1+1/k}\rt)$, where the expectation is taken over $\mu(\mathcal{G}_k)$.

We will show a lower bound for an algorithm that solves the $\nih$ problem. 
Recall that, by the premise of the theorem, we assume a time complexity of $k+1$ for the wake-up algorithm.
According to Lemma~\ref{lem:lb_simul}, we obtain a deterministic algorithm $\mathcal{B}$ for the $\nih$ problem that uses $k+2$ time and an expected message complexity of at most $o\lt(n^{1+1/k}\rt) + n = o\lt(n^{1+1/k}\rt)$, since $k = o\lt( \log n \rt)$.

For each node $v_i \in V$, let $N_i$ denote the set of neighbors of $v_i$, and let $\bar{N}_i = V(G) \setminus N_i$.
Recall from Fact~\ref{lem:graph_props} that $N_i$ is a subset of $U \cup W$.
We define $d=|N_i|=n^{{1}/{k}}+1$. 

\paragraph{Partial ID Assignments and Configurations.}
Since the topology of all graphs in $\mathcal{G}_k$ is fixed, we can uniquely identify each graph $G \in \mathcal{G}_k$ by its assignment of the IDs to the $2n$ nodes in $U \cup W$, which we can describe as a permutation of the set $[2n]$.
From the viewpoint of a node $v_i$, the ID assignment for the nodes in $U \cup W$ can be described as a pair $(I,\rho)$, where $I$ is the ID assignment to the nodes in $\bar{N}_i$ and $\rho$ is the ID assignment to the nodes in $N_i$. 
In our proof, we will first fix $I$ and then consider the possible choices for $\rho$.
This motivates calling $I$ a \emph{partial ID assignment with respect to $v_i$} and $\rho$ a \emph{configuration of $v_i$ with respect to $I$}. 
We write $\rho(u)=x$ to say that node $u$ has ID $x$, and we use $\rho^{-1}(x)$ to denote the neighbor of $v_i$ that $\rho$ assigns ID $x$.
For convenience, we define the notation $\id(S)$ to denote the set of IDs assigned to the nodes in set $S$.
Given a partial ID assignment $I$, we know that the nodes in $N_i$ must be assigned IDs from the set $[2n]\setminus \id(\bar{N}_i)$, and thus $\rho$ can be interpreted as a permutation of $\id(N_i)$.

\begin{fact} \label{cl:num_feasible}
Given a partial ID assignment $I$ with respect to $v_i$, there are $d!=|N_i|! = (n^{1/k}+1)!$ many possibly configurations for $v_i$.
\end{fact}

\paragraph{Quiet and Chatty Configurations.}
Observe that, after fixing a partial ID assignment $I$, the lower bound graph from $\mathcal{G}_k$ is uniquely determined by a configuration $\rho$ of $v_i$; we denote the resulting graph as $G[\rho]$.
We say that a configuration $\rho$ of $v_i$ is \emph{quiet}, if the number of links incident to $v_i$ over which a message is sent when executing on graph $G[\rho]$ is at most $\delta n^{1/k}$, for a positive constant $\delta<0.16$; otherwise we say that $\rho$ is \emph{chatty}.

\begin{lemma} \label{lem:exists_quiet}
There exists a vertex $v_* \in V$ and a partial ID assignment $I$ such that more than a $\frac{7}{8}$-fraction of the configurations of $v_*$ (w.r.t.\ $I$) are quiet.
\end{lemma}
\begin{proof}
Assume towards a contradiction that, for every $v_i \in V$ and every partial ID assignment $I$, at least a $\frac{1}{8}$-fraction of the configurations are chatty. 
Let random variable $M_{v_i}$ denote the total number of messages sent or received by $v_i$, and let $M$ be the total number of messages sent by the algorithm. %
We now compute a lower bound on the expectation of $M$, and we use the notation $\EE_{\rho\mid I}$ to emphasize that the expectation is taken over a uniformly at random sampled configuration $\rho$, conditioned on a given partial ID assignment $I$.
Observe that $2M = \sum_{v_i \in V} M_{v_i}$, since the sum over the $M_{v_i}$ variables counts each message twice, and every edge in the graph has one endpoint in $V$.
Thus, for any partial ID assignment $I$, we have
\begin{align}
\EE_{\rho\mid I}\lt[ M \ \md|\ I \rt] 
  &= \frac{1}{2}\sum_{v_i \in V} \EE_{\rho\mid I}\lt[ M_{v_i}\ \md|\ I \rt] \notag\\ 
  &\ge \frac{1}{2}\sum_{v_i \in V} \EE_{\rho\mid I}\lt[ M_{v_i} \ \md|\ \text{$\rho$ is chatty}, I \rt] \cdot \Pr_{\rho\mid I}\lt[ \text{$\rho$ is chatty} \ \md|\ I \rt]  \notag\\ 
  &\ge \frac{1}{16}\sum_{v_i \in V} \EE_{\rho\mid I}\lt[ M_{v_i} \ \md|\ \text{$\rho$ is chatty}, I \rt]\notag \\ 
  &= \Omega\lt( n^{1+1/k} \rt). \label{eq:exp_messages}
\end{align}
where, in the third step, we have used the assumption that, given $I$, sampling configuration $\rho$ corresponds to sampling a permutation of the IDs in ${N}_i$ uniformly at random and that a $\frac{1}{8}$-fraction of these configurations are chatty.
Moreover, in the fourth step, we used that $\EE_{\rho\mid I}\lt[ M_{v_i} \ \md|\ \text{$\rho$ is chatty}, I \rt] = \Omega\lt( n^{1/k} \rt)$ and that $|V|=n$.
Since the bound in \eqref{eq:exp_messages} holds for every partial ID assignment $I$, it follows that $\EE\lt[ M \rt] = \EE_I\lt[ \EE_{\rho\mid I}\lt[ M \ \md|\ I \rt] \rt] = \Omega\lt( n^{1+1/k} \rt)$,
providing a contradiction to the assumed expected message complexity of $o(n^{1+1/k})$.
\end{proof}

For the remainder of the proof, we consider node $v_*$ and the partial ID assignment $I$, whose existence is guaranteed by Lemma~\ref{lem:exists_quiet}. 
Hence, whenever we specify a configuration $\rho$, we mean a configuration of $v_*$ with respect to $I$. 
Note that we denote $v_*$'s crucial neighbor by $w_*$. 
For a given configuration $\rho$, we define $C_\rho$ to be the set of \emph{communicating neighbors of $v_*$ (in $U$)}, which contains every node $u \in N_* \setminus  \set{w_*}$ such that a message is sent over the edge $\set{u,v_*}$ in $G[\rho]$.
Similarly, we define the \emph{non-communicating neighbors (in $U$)} as $\bar{C}_\rho= N_* \setminus (C_\rho \cup \set{w_*})$. %

In the next lemma, we show that if $v_*$ correctly identifies $w_*$ with ID $x$ in a configuration $\rho$, where some node $u$ and $v_*$ do not communicate, then the algorithm must send a message from $u$ to $v_*$ in the configuration $\rho'$ where we swap the IDs of $w_*$ and $u$. 
Intuitively speaking, this follows since node $v_*$ sees the same neighborhood in both configurations, and the only way $v_*$ can correctly output the ID of $w_*$ is if it receives a message from a node that does see a difference.
This, however, is impossible due to the high girth and the time-restriction, unless the edge $\set{u,v_*}$ is used. 
Figure~\ref{fig:kt1_lb_proof} gives an example of applying Lemma~\ref{lem:swapped}.

\begin{lemma} \label{lem:swapped}
Suppose that $\rho$ is a configuration, where $\rho(w_*) = x$. %
Consider any configuration $\rho'$, with the property that there is some $u \in \bar{C}_\rho$ such that the only difference between $\rho$ and $\rho'$ is that the IDs of $u$ and $w_*$ are swapped. Formally, this means that 
$\rho'(u)=x$,
$\rho'(w_*)=\rho(u)$, and
$\rho'(u') = \rho(u')$, for all $u' \notin \set{u,w_*}$.
Then, it holds that a message is sent from $u$ to $v_*$ over the edge $\set{u,v_*}$ in the execution on $G[\rho']$.
\end{lemma}
\begin{proof}
Since the algorithm is correct, node $v_*$ outputs ID $\rho(w_*)=x$ in $G[\rho]$, but outputs ID $\rho'(w_*)=\rho(u) \ne x$ in $G[\rho']$.
This tells us that $v_*$ must be in a distinct state after some round in the execution on $G[\rho']$ compared to the execution on $G[\rho]$. %
Given that $v_*$ sees the same initial state in both $G[\rho]$ and $G[\rho']$, the only way this can happen is when one of the following two cases occur in $G[\rho']$:
\begin{enumerate} 
\item[(a)] Node $v_*$ receives a message from $u$.
\item[(b)] There exists a chain of messages $\mathcal{C}$  in $G[\rho']$ reaching $v_*$ in some round $r$, where $1 \le r \le k+2$, and $\mathcal{C}$ starts at either $u$ or some neighbor $v'$ of $u$.
See Figure~\ref{fig:kt1_lb_proof} for an example. 
(Note that the neighbors of $u$ see distinct initial states in $G[\rho]$ and $G[\rho']$ due to their $\kt_1$ knowledge.)
\end{enumerate}

In Case~(a), we are done. Thus we focus on Case~(b) and assume that such a chain $\mathcal{C}$ exists.
By Fact~\ref{lem:graph_props}, graph $G$ has a girth of at least $k+5$, and thus the shortest path connecting $v_*$ to $v'$ (or $u$ itself) that does not use the edge $\set{u,v_*}$ must have a length of at least $k+3$.
Recalling that the $\nih$ algorithm $\mathcal{B}$ terminates in at most $k+2$ rounds, implies that $\mathcal{C}$ must traverse $\set{u,v_*}$, which completes the proof.
\end{proof}

\begin{figure}[t]
  \centering
\includegraphics[scale=1.25]{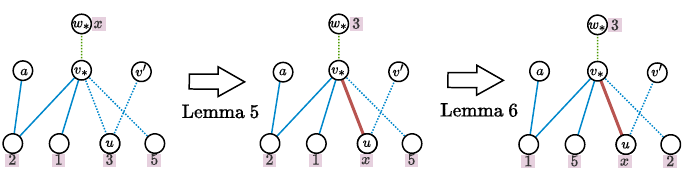}
\caption{\small
Three neighborhood ID assignments (i.e., configurations) of node $v_*$.
Shaded boxes depict the node IDs. 
Dashed lines correspond to edges that are not used for communication.
Applying Lemma~\ref{lem:swapped} to the left configuration tells us that $u$ must send a message to $v_*$ in the middle configuration, where the IDs of $u$ and $w_*$ are swapped.
While node $v'$ sees distinct neighborhoods in these configurations due to the $\kt_1$ assumption, it cannot inform $v_*$ about this without using the edge $\set{u,v_*}$ due to the time-restriction.
Moreover, Lemma~\ref{lem:always_send} ensures that a message is sent over the edge $\set{u,v_*}$ in the configuration on the right, where $u$ has the same ID as in the middle configuration.
Again, this is because any node $a$ who observes a difference between the two configurations cannot convey this to $u$ on time, due to the high girth of the graph.
}
\label{fig:kt1_lb_proof}
\end{figure}

Next, we show that if a neighbor $u$ with a given ID $x$ sends a message to $v_*$, then some message is sent over this edge in all configurations where $u$ has ID $x$.
See Figure~\ref{fig:kt1_lb_proof} for an example.

\begin{lemma} \label{lem:always_send}
Consider a configuration $\rho$ and suppose a neighbor $u$ with ID $x$ sends a message to $v_*$ in configuration $G[\rho]$.
Then, for any configuration $\rho'$, where $u$ has ID $x$, it holds that some message is sent over $\set{u,v_*}$ in $G[\rho']$.
\end{lemma}
\begin{proof}
We will show that $u$ has the same state in every round in both $G[\rho]$ and $G[\rho']$, as this implies that it will also send a message to $v_*$ in both executions.

The initial state of $u$ depends only on its neighborhood which is the same in $G[\rho]$ and $G[\rho']$. Now, assume towards a contradiction that there exists some round $r \ge 1$, at the end of which $u$'s state differs in $G[\rho]$ from the one it had in $G[\rho']$; without loss of generality, assume that $r$ is the earliest round in which this is the case.
Thus, there must be some chain of messages $\mathcal{C}$ starting at some node $a \in U \cup V$ and ending at $u$ in round $r$, such that node $a$ observes distinct neighborhoods in $G[\rho']$ compared to $G[\rho]$.
Since the only difference between $G[\rho]$ and $G[\rho']$ is the ID assignment to $v_*$'s neighbors, it follows that $a$ must either itself be in $N_*$ or a neighbor of some node in $N_*$ (which may suffice for seeing a different neighborhood due to $\kt_1$).
We can assume that $v_*$ does not send a message to $u$ in the execution on $G[\rho']$, as otherwise we are done.
By Fact~\ref{lem:graph_props}, graph $G$ has a girth of at least $k+5$, and thus any shortest path between $a$ and $u$ that does not use the edge $\set{u,v_*}$ must have a length of at least $k+3$.
However, the existence of such a chain is impossible since the algorithm must terminate in $k+2$ time units.
Thus, it follows that $u$ has the same state in every round in both $G[\rho]$ and $G[\rho']$, which contradicts the existence of $r$.
\end{proof}

We now show that there is a large set of IDs $X$ such that there is a quiet configuration in which $w_*$ has ID $x$, for each $x \in X$.
We will use this in the proof of Lemma~\ref{lem:too_many_chatty} in combination with Lemmas~\ref{lem:swapped} and \ref{lem:always_send} to show that $v_*$ does indeed have many chatty configurations.

\begin{lemma} \label{lem:chatty_crucial}
There exists a set of IDs $X$ of size at least $\frac{7}{8}|N_*|$ such that, for each $x \in X$, there is a quiet configuration $\rho_x$ of $v_*$ where $\rho_x(w_*) = x$.
\end{lemma}
\begin{proof}
Suppose that $|X| < \frac{7}{8}|N_*|$ and note that, for every $z \notin X$, every configuration where $\id(w_*) = z$ must be chatty.
This means that there are at least 
\begin{align}
  (|N_*|-|X|)(d-1)! \geq \frac{1}{8}|N_*|(d-1)! = \frac{(n^{1/k}+1)!}{8} 
\end{align}
many chatty configurations of $v_*$. 
However, according to Lemma~\ref{lem:exists_quiet} and Fact~\ref{cl:num_feasible}, strictly less than $\frac{1}{8} (n^{1/k}+1)!$ configurations can be chatty, providing a contradiction.
\end{proof}
 
\begin{lemma} \label{lem:too_many_chatty}
Given the partial ID assignment $I$, a uniformly at random sampled configuration for $v_*$ is chatty with probability at least $\frac{1}{4}$. 
\end{lemma}
\begin{proof}
Let $X$ be the set of IDs guaranteed by Lemma~\ref{lem:chatty_crucial}, and recall that $\rho_x$ refers to the corresponding quiet configuration for $x \in X$.
Now, consider any configuration $\rho$ such that there exists some $u \in \bar{C}_{\rho_x}$ with $\rho(u) = x$.
Informally, this means that $\rho$ can be any configuration where one of the neighbors of $v_*$ that were non-communicating in $G[\rho_x]$ has ID $x$ in $G[\rho]$.
Combining Lemmas~\ref{lem:swapped} and \ref{lem:always_send} tells us that a message is sent on the edge connecting $v_*$ to the node with ID $x$ in $G[\rho]$.
We now use this observation to derive a lower bound on the probability that there is a message sent between the node with ID $x$ and $v_*$ in the execution on $G[\rho]$, where configuration $\rho$ is sampled uniformly at random,  conditioned on the partial ID assignment $I$.
We obtain
\begin{align}
\Pr_{\rho\mid I}\lt[ \text{$\exists u \in \bar{C}_{\rho_x}$: $\rho(u)\!=\! x$} \ \md|\ I\rt]
= 
\frac{|\bar{C}_{\rho_x}|}{|N_*|} \ge \frac{(1-\delta)n^{1/k}}{n^{1/k}+1}\ge (1-\delta)^{2}. \label{eq:prob_id_bnd}
\end{align}
In \eqref{eq:prob_id_bnd}, we used the fact that $\rho_x$ is quiet, thus $|\bar{C}_{\rho_x}| \ge (1-\delta)n^{1/k}$, and we also observed that $\frac{n^{1/k}}{n^{1/k}+1} \ge 1-\delta$, for sufficiently large $n$.

Let $T$ denote the number of neighbors $u$ of $v_*$, such that a message is communicated over $\set{u,v_*}$, and let $T_x$ be the indicator random variable that is $1$ if and only if a message is sent over the edge between $v_*$ and the node with ID $x$.
By linearity of expectation, we have 
\begin{align}
  \EE_{\rho \mid I}\lt[ T \ \md|\ I\,\rt] 
  &=
  \sum_{x \in \id(N_*)} \EE_{\rho \mid I}\lt[ T_x  \ \md|\ I\,\rt]  \notag\\ 
  &\ge \sum_{x \in X} \Pr_{\rho\mid I}\lt[ \text{$\exists u \in \bar{C}_{\rho_x}$: $\rho(u)=x$}, \ \md|\ I\, \rt]  \notag \\
  \ann{by \eqref{eq:prob_id_bnd}}
  &\ge |X| \cdot (1-\delta)^2 \notag\\ 
  &\ge \frac{7}{8}(1-\delta)^{2}n^{1/k}. \label{eq:prob_bnd2}
\end{align}
Consequently, we get
\begin{align}
  \Pr_{\rho \mid I}\lt[ \text{$\rho$ is quiet}\ \big|\ I \rt] 
  &=
  \Pr_{\rho \mid I}\lt[ T \le \delta n^{1/k}\ \big|\ I \rt] \notag\\ 
  &=
  \Pr_{\rho \mid I}\lt[ |N_*| - T \ge |N_*| - \delta n^{1/k}\ \big|\ I \rt] \notag\\ 
  \ann{by Markov's inequality}
  &\le \frac{\EE\lt[ |N_*| - T \ \md|\ I\, \rt] }{|N_*|-\delta n^{1/k}}\notag \\
	\ann{since $|N_*|=n^{1/k}+1$ and by \eqref{eq:prob_bnd2}}
  &\le \frac{(n^{1/k}+1)-\frac{7}{8}(1-\delta)^2n^{1/k}}{n^{1/k}-\delta n^{1/k}}\notag \\
  &\le \frac{1-\frac{7}{8}(1-\delta)^2}{1-\delta} + o(1)\notag \\
  \ann{since $\delta<0.16$}
  &\le \frac{1}{2} + o(1) \le \frac{3}{4},\notag
\end{align}
where we have applied Markov's inequality for the random variable $|N_*|-T$ in the third line.
\end{proof}

Lemma~\ref{lem:too_many_chatty} provides a contradiction to Lemma~\ref{lem:exists_quiet}, which confirms that any deterministic algorithm has an expected message complexity of $\Omega\lt( n^{1+1/k} \rt)$ on the uniform distribution $\mu(\mathcal{G}_k)$.
To obtain the analogous lower bound on the expected message complexity for randomized Las Vegas algorithms follows, we apply Yao's Minimax Lemma~\cite{yao1977probabilistic}. 
This completes the proof of Theorem~\ref{thm:lb}.

%% file: kt1_algorithm.tex
\section{Algorithms in the $\kt_1$ $\local$ Model} \label{sec:kt1_algos}
In this section, we describe how to achieve wake-up by leveraging the $\kt_1$ assumption and the $\local$ model, where nodes are allowed to send messages of arbitrary size.

\subsection{Asynchronous Wake-Up with Near-Optimal Message Complexity} \label{sec:kt1_algo_async}

Let $A$ be the nodes that are awakened by the adversary. 
Every node $u \in A$ chooses a random rank $\rho_u$ from the range $[n^{c}]$, where $c$ is a suitable large constant.
(This is possible because each node knows a constant factor upper bound on $\log n$, see Sec.~\ref{sec:prelim}.) 
Then $u$ starts a depth-first-search (DFS) traversal of the network via token passing, where the token carries $u$'s rank $\rho_u$, $u$'s ID, and the full list of IDs of the nodes the token has visited so far.
We remark that nodes who are awoken by a received message (instead of the adversary) neither initiate a DFS traversal nor create a random rank. 

An awake node $v$ may be visited by DFS traversal tokens originating from distinct sources, and it keeps track of the largest rank of any token that it has seen in a variable $\rho^{*}$, and stores the corresponding origin node's ID in $\id^{*}$. 
Whenever $v$ receives a new token with a rank $\rho_u$ originated from node $u$ with ID $\id(u)$, it lexicographically compares $(\rho_u,\id(u))$ and $(\rho^{*},\id^{*})$: 
\begin{enumerate}
\item[(a)] If $(\rho_u,\id(u)) > (\rho^{*},\id^{*})$, then $v$ appends its own ID to the list $l$ of IDs carried by the token and it allows the token to continue by passing it to some neighbor, whose ID is not yet in $l$. 
Moreover, $v$ updates its variables $\rho^{*}$ and $\id^{*}$ accordingly.
If no such neighbor is found, we backtrack by returning the appended token back to its DFS-parent, who will in-turn forward the token to another not-yet-visited neighbor  (if any) in a DFS-manner, and so forth.  
\item[(b)] Otherwise, if $(\rho_u,\id(u)) < (\rho^{*},\id^{*})$, then node $v$ simply discards this token. 
\end{enumerate}

\begin{reptheorem}{thm:dfs_many} %
  \thmDFSMany
\end{reptheorem}
\onlyLong{
\subsubsection{Proof of Theorem~\ref{thm:dfs_many}}
Since the list of already-visited IDs carried in the token is used for determining which neighbor to forward it to, it is easy to see that, for a given token $\pi$: (1) the path taken by $\pi$ does not form a cycle, and (2) $\pi$ traverses each edge on the path at most twice.
This implies the following:
\begin{claim} \label{cl:token}
The path traversed by any token is a tree and the number of times a given token is forwarded is $O(n)$.
\end{claim}

We first argue the correctness.
Let $r^{*}$ be the maximum rank that is generated by any node, and let $S$ denote the set of nodes that chose $r^{*}$.
Moreover, let $v^{*} \in S$ be the node with the largest ID in $S$.
Since we use lexicographic comparisons, it is straightforward to verify that $v^{*}$'s DFS-token is never discarded, i.e., case (a) occurs for every node that it visits, which guarantees that every node is woken up eventually with probability $1$, thus showing correctness.

The following claim is a byproduct of the above argument:
\begin{claim} \label{cl:wakeup}
Consider any set of nodes $S$ that are woken up by the adversary at time $t$, and suppose that the token $\pi$ originating from the node with the maximum rank in $S$ does not encounter any token originating from a node outside of $S$. Then, $\pi$ wakes up every node by time $t+2n$.
\end{claim}
The proof of Claim~\ref{cl:wakeup} is immediate because some node $v \in S$ must have sampled the maximum rank. 
As argued above, $v$'s token will never be discarded and visits every node in the network, while traversing each edge of the resulting DFS tree at most twice.

Next we show the claimed bound on the time complexity.
Recall that each node that is awakened by the adversary chooses its rank from $[n^{c}]$ uniformly at random, where $c$ is a sufficiently large constant.
By taking a union bound over all nodes, it is straightforward to see that all ranks are unique w.h.p, 
and we will silently condition on this event throughout the rest of the proof.

The next claim follows from the fact that the adversary is oblivious to the ranks. 
\begin{claim} \label{cl:rank}
Let $S$ be a set of nodes. Given a node $v \in S$, the probability that $v$ has the highest rank among all nodes in $S$ is $\frac{1}{|S|}$. 
\end{claim} 

Let $t_0$ be the earliest time that the adversary wakes up some set of nodes $S_0$, and let $v_0 \in S_0$ be the node that samples the maximum rank among nodes in $S_0$; recall that the adversary is oblivious to the chosen ranks.
If $v_0$'s token does not encounter any tokens from nodes outside of $S_0$, then  Claim~\ref{cl:wakeup} ensures every node wakes up by time $t_0+2n$. 
To prevent this from happening, the adversary must wake up another nonempty set of nodes $S_1$ (disjoint from $S_0$) at some time $t_1<t_0+2n$, i.e., each node in $S_1$ has its token in transit by time $t_1$.
According to our algorithm, the only way that $v_0$'s token is discarded is if the rank of some node $v_1 \in S_1$ is greater than $\rho_{v_0}$.
By Claim~\ref{cl:rank} this happens with some probability at most $\gamma_1=\min\Set{1,\frac{|S_1|}{|S_0|}}$.

We now apply Claim~\ref{cl:wakeup}, this time to the set $S_0 \cup S_1$ and use $v_1$ to denote the node who produced the maximum rank in $S_0 \cup S_1$. Again, this tells us that the adversary must wake up nodes in a disjoint set $S_2$ such that $v_1$'s token was discarded before time $t_1+2n$, and analogously to above, the probability that this happens is at most
$\gamma_2=\min\Set{1,\frac{|S_2|}{|S_0 \cup S_1|}}$, and so forth.
This yields sequences $(S_i)_{i\ge0}$ and $(\gamma_i)_{i\ge0}$ such that, for every $i \ge 1$,
\begin{align} \label{eq:size_gamma}
\gamma_i = \min\Set{1,\frac{|S_i|}{\big|\bigcup_{j=0}^{i-1}S_j\big|}}
= \min\Set{1,\frac{|S_i|}{\sum_{j=0}^{i-1}|S_j|}}.
\end{align}

Now suppose the adversary tries to repeat this process of stopping the token with the current maximum rank $k$ times where $k \ge 8 \log n$ \footnote{All logarithms are assumed to be of base $e$ unless stated otherwise.}.
Our goal is to show that this happens with probability at most $\frac{1}{n^{2}}$.
Let $I \subseteq \set{1,\ldots,k}$ be the set of indices such that 
\begin{align} 
\forall i \in I\colon \gamma_i \ge 2 - \sqrt{e}. \label{eq:gamma}
\end{align}

To simplify the presentation, we assume without loss of generality that $I = \set{1,\dots,|I|}$, and
we also define the partial sum $P_\ell = \sum_{i=0}^{\ell}|S_i|$.
For every $i \in I$, \eqref{eq:size_gamma} and \eqref{eq:gamma} imply that
$|S_i| \ge \lt(2-\sqrt{e}\rt) P_{i-1},
$
and thus
\begin{align}
\forall i \in I\colon P_{i} &\ge \lt(3-\sqrt{e}\rt) P_{i-1}. \label{eq:sumlb}
\end{align}
Applying \eqref{eq:sumlb} to each of the partial sums for $i=1,\ldots,|I|$, and recalling that $|S_0| \ge 1$, yields
\begin{align}
P_{|I|} &\ge \lt(3-\sqrt{e}\rt)^{|I|}. 
\end{align}
Since $S_1,\dots,S_i$ are pairwise disjoint it holds that $P_{|I|} \le n$. 
Consequently, it must be that $|I| < 4\log n$, as otherwise the right-hand side is at least $n^{4 \log(3-\sqrt{e})} > n$.
In other words, \eqref{eq:gamma} holds for less than $4\log n$ indices out of the $k\ge 8\log n$, and hence there is a set of indices $J \subseteq [k]$ of size 
\begin{align}
|J| = k - |I| \ge 4\log n,  \label{eq:size_j}
\end{align}
such that $\gamma_j < 2 - \sqrt{e}$, for every $j \in J$.

Let $E_i$ be the event that the maximum rank of nodes in $S_i$ is greater than the maximum rank in $\bigcup_{j=0}^{i-1}S_j$. 
We point out that $\bigwedge_{i=1}^{k} E_i$ is a necessary condition for the adversary to be able to repeat the process of successfully prolonging the termination time by discarding the current maximum rank $k$ times.
From the above, we know that $\Pr\lt[ E_j \rt] \le \gamma_j < 2-\sqrt{e}$, for each $j \in J$, and it is straightforward to see that the events $(E_j)_{j \in |J|}$ are mutually independent.
Combining these observations, we have
\begin{align}
\Pr\lt[ \bigwedge\nolimits_{i=1}^{k} E_i \rt] 
  \le
\Pr\lt[  \bigwedge\nolimits_{j\in J} E_j\rt] 
=
\prod_{j\in J}
\Pr\lt[  E_j\rt] 
\le
\lt(2-\sqrt{e}\rt)^{|J|}
\le
\frac{1}{n^2},
\end{align}
where we have used \eqref{eq:size_j} in the final inequality.
We conclude that the time complexity is $O\lt( n\log n \rt)$ with high probability.

Finally, we show the claimed bound on the message complexity.
For a token $\pi$, we use $P_\pi$ to denote the (directed) path taken by $\pi$ and  $|P_{\pi}|$ to denote the number of edges in $P_\pi$. 
Hence, the number of messages sent, denoted as $M$, is equal to $\sum_{\pi}|P_{\pi}|$ where the sum is taken over all tokens generated in the execution.  
From Claim~\ref{cl:token}, we know that $P_\pi$ is a DFS traversal of a tree. 
This implies there are at least $\frac{|P_\pi|}{2}$ distinct nodes visited by $\pi$\footnote{This follows since a tree of $t$ nodes has $t-1$ edges and a path corresponding to a DFS traversal of this tree has $2(t-1)$ edges.}.
With the exception of possibly one such node that may have discarded the token, the rest of these nodes in $P_\pi$ forwarded $\pi$. 
Hence, the number of distinct nodes that forward $\pi$, denoted as $\nu(\pi)$, is at least $\frac{|P_\pi|}{2}-1$. 
As a result, we have 
\[
M = \sum\nolimits_{\pi}|P_{\pi}| \leq \sum\nolimits_{\pi} (2\nu(\pi)+1) = O\lt(\sum\nolimits_{\pi} \nu(\pi)\rt). 
\]
To prove the message complexity bound, we aim to show that $\hat{M}:=\sum_{\pi} \nu(\pi)$ is $O(n \log n)$ with high probability.
Let $\tau(v)$ denote the number of distinct tokens that were forwarded by node $v$. 
It is straightforward to see that $\hat{M}$ is equivalent to $\sum_{v} \tau(v)$ where the summation is taken over all nodes. 
Thus, it remains to show that $\tau(v)$ is $O(\log n)$ w.h.p for each node $v$, which implies the sought upper bound on the message complexity.
\begin{claim} \label{cl:forward}
Each node forwards $O(\log n)$ distinct tokens w.h.p.
\end{claim} 
The proof of this claim follows from Claim~\ref{cl:rank} and a standard Chernoff bound, along similar lines as the analysis for ``least element lists'' (see \cite{cohen1997size,DBLP:journals/dc/KhanKMPT12}). 
We include a full argument in Appendix~\ref{app:kt1_algo} for completeness.
This completes the proof of Theorem~\ref{thm:dfs_many}.
}

\subsection{Achieving Wake-Up in $\tilde O(\arad)$ Synchronous Rounds} \label{sec:kt1_algo_sync}
 
\input{kt1_sync}

%% file: kt1_sync.tex
We now describe an algorithm that achieves an (optimal) time complexity that asymptotically matches the awake distance, while sending only $o(m)$ messages in sufficiently dense graphs.
Here, we assume the synchronous model, where the computation is structured into lock-step rounds. 
Every node performs a local computation step at the start of each round $r$, and decides which messages to send.
In particular, any message sent in round $r$ is guaranteed to be delivered by the start of round $r+1$.
Just as in the asynchronous model, the adversary determines the set of initially-awake nodes and also determines, for all other nodes the round number in which they wake up, unless they are already awoken earlier due to receiving a message. 
We emphasize that nodes do \emph{not} have access to a global clock, i.e., when a node is awoken by the adversary, it does not know how many rounds have passed since the first node started executing the algorithm.

\subsubsection{Algorithm $\fastwakeup$}
Each node can either be in one of the following states, namely \emph{asleep}, \emph{passive}, \emph{activated}, and \emph{deactivated}, whereby only awake nodes can be in either of the latter three states.
Moreover, once a node becomes \emph{deactivated}, it remains in this state until termination.
Every node that is awoken by the adversary marks itself as \emph{activated}, which, intuitively speaking means that it will either initiate the construction of a BFS tree or simply broadcast a message to all its neighbors at some point unless it becomes \emph{deactivated} before doing so.
Then, every activated node $u$ starts executing the following procedure: 
\begin{itemize} 
\item \textbf{Sampling Step:} Node $u$ becomes a \emph{root} with probability $\sqrt{\frac{\log n}{n}}$. Note that $u$ remains \emph{activated} even if it is not sampled.
\item \textbf{BFS Tree Construction:}
Each root $u$ initiates the construction of a BFS tree of depth $3$, by using the following ``back-and-forth'' technique of \cite{opodis2024}: $u$ sends a message to all its neighbors.  Each neighbor joins $u$'s tree as a \emph{level-$1$ node} and, in return, sends $u$ the list of IDs of its own neighbors. 
Upon receipt of these lists from all neighbors, $u$ can locally compute a set of BFS-tree edges, called $S_2$, that connect the first and second level of its BFS tree, which it sends to its neighbors. 
Node $u$ also attaches its own ID (for distinguishing messages from different roots) and a distance counter to each of these \emph{BFS tree construction messages}.
The latter is incremented on each tree level as the message traverses the tree.
In the next round, the level-$1$ nodes of $u$'s tree send messages across their incident edges in $S_2$ that they have learned from $u$. %
As a result, each distance-2 neighbor $v$ of $u$ will join the tree as a \emph{level-$2$ node}.
To determine the next layer of the tree, we proceed analogously as for the first level: 
Node $v$ sends the list of its neighbors' IDs via its parent up to the root $u$, thus enabling $u$ to locally compute a set of BFS-tree edges, called $S_3$, between its second and third tree level.
The BFS tree construction stops after the level-$2$ nodes have sent messages over their incident BFS-edges (in $S_3$) to their respective children on level 3, thereby recruiting the \emph{level-$3$ nodes} to $u$'s tree. %

An \emph{asleep} node that wakes up upon receiving a BFS construction message becomes \emph{passive}.

At the end of this BFS tree construction, the root and all level-$1$ and level-$2$ nodes  change their status to \emph{deactivated}.
On the other hand, any level-$3$ node that was not yet \emph{deactivated} becomes \emph{activated}. As mentioned above, an already-\emph{deactivated} node never becomes \emph{activated} again, and thus remains deactivated even after joining a tree on level $3$. 
A node $w$ may participate in the construction of multiple BFS trees simultaneously. In that case, $w$ simply performs the above procedure in parallel for all such trees, which is possible without slowdown in the $\local$ model.
Note that $w$ becomes \emph{deactivated} if it is on level $1$ or level $2$ in at least one of these tree.

\item \textbf{Broadcast Step:} If a node has remained \emph{activated} for $9$ rounds, then it broadcasts an $\msg{\texttt{activate!}}$ message to its neighbors at the start of its $10^{\text{th}}$ round and becomes \emph{deactivated}.
Any node $w$ that is awoken by an $\msg{\texttt{activate!}}$ message immediately becomes \emph{activated} and tries to sample itself as a root by following the procedure described above. 
If $w$ simultaneously receives BFS tree construction messages, then it simply ignores the $\msg{\texttt{activate!}}$ message and joins the corresponding BFS trees accordingly.

\end{itemize}

\begin{reptheorem}{thm:kt1_fw}
\thmFW
\end{reptheorem}
\subsubsection{Proof of Theorem~\ref{thm:kt1_fw}}

The following lemmas are straightforward from the description of the algorithm. We postpone the proofs to Appendix~\ref{app:kt1_algo_sync}.
\newcommand{\lemDeactivated}{
If a node $u$ changes its status to deactivated in some round $r$, then every neighbor of $u$ is awake at the start of round $r$. 
}
\begin{lemma} \label{lem:deactivated}
\lemDeactivated
\end{lemma}

\newcommand{\lemRound}{
The BFS tree construction procedure requires $9$ rounds. 
}
\begin{lemma} \label{lem:9round}
\lemRound
\end{lemma}

\newcommand{\lemActive}{
If a node wakes up in round $r$, then it becomes deactivated at the latest at the end of round $r+10$.
}
\begin{lemma} \label{lem:active}
\lemActive
\end{lemma}

\begin{lemma} \label{lem:fwTime}
Algorithm $\fastwakeup$ wakes up all nodes in $10\arad$ rounds.
\end{lemma}
\begin{proof}
Note that every node $u$ that is initially asleep must have an initially awake node $v$ within at most $\arad$ hops. 
It follows from Lemma~\ref{lem:active} and Lemma~\ref{lem:deactivated} that after at most $10$ rounds, all of $v$'s neighbors are awake. 
Hence, after at most $10\arad$ rounds, $u$ is awake. 
Finally, observe that nodes that are awoken by the adversary at a later point become active, which, however, does not slow down any BFS constructions of the initially-awake nodes that are already in-progress.
\end{proof}

\begin{lemma} \label{lem:fwMessages}
The message complexity of algorithm $\fastwakeup$ is $O\lt(n^{3/2}\sqrt{\log n}  \rt)$ with high probability.
\end{lemma}

\begin{proof}
We separately derive the bound for the two types of messages sent by the algorithm, namely, the BFS construction messages and the $\msg{\texttt{activate!}}$ messages sent during the broadcast step. 

We first argue that at most $O\lt( n^{3/2}\sqrt{\log n} \rt)$ messages are sent by the nodes constructing the BFS trees. 
Since the number of \emph{activated} nodes is bounded by $n$, at most $O\lt(  \sqrt{n\log n}\rt)$ of them become roots, in expectation, which is an upper bound on the expected number of distinct trees constructed throughout the execution of the algorithm. 
Note that these messages are only sent along the BFS-edges, and, for constructing a new level, a message is sent over each BFS-edges at most three times.
Thus, we can upper-bound the expected number of messages for constructing a single BFS tree by $O\lt( n \rt)$. 
This yields $O\lt( n^{3/2}\sqrt{\log n} \rt)$ expected BFS construction messages overall.   
Since each node is sampled as a tree root independently, a standard Chernoff bound argument shows that this also happens with high probability. 

Next, we need to bound the number of $\msg{\texttt{activate!}}$ messages that are sent throughout the execution of the algorithm.  
We aim to show that any node $v$ receives less than $10\sqrt{n \log n}$ $\msg{\texttt{activate!}}$ messages with high probability. 
Then, by taking a union bound over all nodes, it follows that there are at most $O\lt( n^{3/2}\sqrt{\log n} \rt)$ $\msg{\texttt{activate!}}$ messages in total.

Consider any node $v$ that has $\ell$ neighbors who execute the sampling step, and suppose that $\ell \ge 8 \sqrt{n \log n} + 1$, as otherwise we are done. 
Order these $\ell$ neighbors according to the earliest round in which they become \emph{activated}. 
Define $L$ to be the set containing the first $8 \sqrt{n \log n}$ nodes\footnote{For simplicity, we treat $\sqrt{n \log n}$ as an integer.} in this order, and let $t_*$ be the latest time at which a node in $L$ becomes \emph{activated}.
We use $\bar{L}$ to denote the remaining $\ell - |L| \ge 1$ neighbors of $v$ which execute the sampling step.
Recall that, by the sampling step, each node in $L$ tries to become root with probability $\sqrt{\frac{\log n}{n}}$.
Thus, the event $\bad$, which we define to occur if none of the nodes in $L$ becomes root, happens with probability at most 
\[
\lt( 1 -  \sqrt{\frac{\log n}{n}}\rt)^{|L|}
\le
\lt( 1 -  \sqrt{\frac{\log n}{n}}\rt)^{8 \sqrt{n \log n}}
\le 
\frac{1}{n^{4}}.
\]
Conditioned on event $\neg \bad$, let $w \in L$ be a neighbor of $v$ who became a root and note that this must happen at the latest by round $t_*$. 
It follows that $w$ initiates the construction of a BFS tree, which completes in round $t_*+9$ by Lemma~\ref{lem:9round}.
Recall that any $x \in \bar{L}$ becomes \emph{activated} in some round $t' \ge t_*$.
Moreover, $x$ has a distance of at most $2$ from $w$, and thus it follows that $x$ will join $w$'s tree as a level-$1$ or level-$2$ node, causing it to become deactivated before it can execute the broadcast step.
Consequently, $v$ does not receive any $\msg{\texttt{activate!}}$ messages from the nodes in $\bar{L}$.
Since event $\neg \bad$ happens with high probability, we can remove the conditioning to conclude that $v$ receives at most $|L|=O\lt( \sqrt{n\log n} \rt)$ $\msg{\texttt{activate!}}$ messages, as required.
\end{proof}

\subsubsection{Algorithm $\pfw$} \label{sec:pfw}

Algorithm~$\fastwakeup$ can be viewed as the final two phases of a more general algorithm, called $\pfw$, which performs $\Theta\lt( \log n \rt)$ phases of BFS tree explorations of decreasing depth. 
Algorithm~$\pfw$ is similar in spirit to the sparse cover construction of Elkin~\cite{elkin2006faster} that was originally designed and analyzed for the synchronous $\kt_0$ $\congest$ model, where all nodes wake-up simultaneously.

Recall that every node that is awoken by the adversary marks itself as \emph{activated}. 
Each \emph{activated} node $u$ performs $K=\Theta\lt( \log n \rt)$ phases.
Upon entering phase $p=1,\dots,K-1$, node $u$ samples itself as a \emph{phase-$p$ root} with probability $\gamma_p = \frac{\sqrt{\log n}}{n^{1-p/K}}$.
If $u$ was not sampled in phase $p \le K-2$, it waits $\Theta\lt( d_p^2 \rt)$ rounds, where $$
d_p=2+K-p,$$ before sampling itself for the next phase, by increasing the probability accordingly.
On the other hand, if $u$ is successfully sampled as a phase-$p$ root, it initiates the construction of a BSF tree of depth $d_p$
analogously to the message-efficient ``back-and-forth'' technique described in algorithm $\fastwakeup$, i.e., by iteratively convergecasting and broadcasting the entire neighborhood information between the (current) frontier nodes and the root.
Every node $v$ that joins $u$'s tree on a level $l \le d_p-1$  in some round $r$, changes its status to \emph{deactivated} upon the completion of all tree constructions that it is currently participating in.
Moreover, we say that $v$ becomes part of the \emph{kernel} of $u$'s tree at the start of round $r+1$.
Note that the root itself also becomes \emph{deactivated} upon completion.
On the other hand, if a node joins the tree on level $d_p$ (i.e., as a leaf), then it changes its status to \emph{activated} (or keeps its current \emph{activated} status). 
If, at the start of round $r$, a node $w$ knows that its neighbor $v$ is deactivated, then it also deactivates the edge $\set{w,v}$, which means that $w$ does not send any message across this edge in any subsequent round $r' \ge r$.

If a node remains active throughout phases $1,\dots,K-1$ and did not become a root in any of them, then it enters the final phase $p=K$, where it simply broadcasts an  $\msg{\texttt{activate!}}$ message to all its neighbors and deactivates itself, analogously to algorithm $\fastwakeup$. 

Finally, we provide more details regarding how the status of a node may change upon wake-up and the receipt of messages.
First of all, we point out that a \emph{passive} node can only become \emph{deactivated}, and a \emph{deactivated} node never changes its status again. 
If a node $v$ is woken up by the adversary and does not receive any simultaneous messages, its status is \emph{activated}.
If $v$ is woken up by a tree message that causes it to join the tree as a leaf, then it also becomes \emph{activated}, otherwise it is part of the tree's kernel and becomes \emph{passive}.
It can happen that $v$ receives messages from multiple distinct trees (even in the same round), and it may join some tree $T$ as leaf and join the kernel of another tree $T'$, i.e., at some depth less than $d_p$.
In that case, $v$ simply becomes \emph{passive}.
The intuition behind this rule is that $v$ knows that all its neighbors will join $T'$ (and hence are woken up), and thus there is no reason for $v$ itself to become \emph{active}.
When $v$ is awoke due to receiving a tree as well as some $\msg{\texttt{activate!}}$ messages, then it discards the $\msg{\texttt{activate!}}$ messages and updates its status according to the rules described above.

For our analysis, we say that a message is a \emph{phase-$p$ tree message} if the message is sent due to the construction of a BFS tree rooted at a node that was sampled as a root in phase $p$.

\begin{lemma} \label{lem:deact_no_msg}
If a node $w$ becomes \emph{deactivated} in some round, then $w$ does not receive any message in any later round.
\end{lemma}
\begin{proof}
We need to show that each awake and non-deactivated node $v$ knows, at the start of round $r$, which of its neighbors have become deactivated at the end of round $r-1$.
We perform a case distinction over the possible ways that some neighbor $w$ of $v$ may become deactivated.
Clearly, $v$ knows that $w$ is deactivated if it has sent an $\msg{\texttt{activate!}}$ message in round $r-1$.
It remains to show the statement for the case where $w$ deactivates itself after the completion of the construction of some phase-$p$ BFS tree $T$ rooted at some node $u$, of which it is part of the kernel, as leaves remain activated.
Since $w$ has a distance of at most $d_p-1$ from $u$, it follows that all (non-deactivated) neighbors also join $T$.
The statement is immediate if $v$ is either $w$'s parent or its child.
To see why the lemma holds in the case where $\set{w,v}$ is not an edge of $T$, recall that each node in the tree learns the entire topology. 
Since we assume $\kt_1$, this is sufficient for $v$ to deduce which of its neighbors are in $T$ as well as their distance from the root. 
Thus, $v$ knows that $w$ is not a leaf, and hence will deactivate the edge $\set{w,v}$ accordingly.
\end{proof}

\begin{lemma} \label{lem:pfw_bfs_time}
The construction of a BFS tree initiated by a phase-$p$ root takes $O(d_p^2) = O(\log^2 n)$ rounds.
\end{lemma}
\begin{proof}
It is straightforward to see that when the BFS tree has been constructed up to some depth $d$, it will take $O(d)$ additional rounds to extend the BFS tree to depth $d+1$, according to the description of the algorithm. Hence, constructing the entire tree up to  depth $d_p$ requires $O(d_p^2)$ rounds. The lemma follows since $d_1 = \Theta\lt( \log n \rt)$ and $d_1 > \dots > d_{K-1}$. 
\end{proof}

\begin{lemma} \label{lem:pfw_deact_time}
If a node wakes up in round $r$, then it becomes \emph{deactivated} at the latest after $O(\log^3 n)$ rounds.
\end{lemma}
\begin{proof}
When a node is awake and is \emph{passive}, then it has joined a tree and will become deactivated upon the completion of the tree construction. 
When a node is awake and is \emph{activated}, it executes at most $K$ phases before becoming \emph{deactivated}. In each phase $p$, it either constructs a BFS tree or it waits $\Theta(d_p^2)$ rounds before proceeding to the next phase if it has not yet been deactivated. The statement follows from Lemma~\ref{lem:pfw_bfs_time} and since $K=\Theta(\log n)$. 
\end{proof}

\begin{lemma} \label{lem:pfw_time}
Algorithm $\pfw$ wakes up all nodes in $O\lt( \arad \log^3n \rt)$ rounds.
\end{lemma}
\begin{proof}
Note that when a node is deactivated at round $r$, all its neighbors are awake by the end of $r+1$. 
Let $u$ be an \emph{asleep} node. Then $u$ has an initially awake node $v$ within its $\arad$ neighborhood. From Lemma~\ref{lem:pfw_deact_time}, it takes $r = O(\log^3n)$ rounds for $v$ to be deactivated, and all its neighbors are awake by round $r+1$.
Hence, after at most $O(\arad \log^3n)$ rounds, $u$ is awake. 
\end{proof}

\begin{lemma} \label{lem:pfw_tree_msgs}
With high probability, the following hold:
Each node $v$ receives tree messages from at most $O(\log n)$ distinct phase-$p$ trees, for each $p \in [K]$.
Moreover, $v$ receives at most $O(\log n)$ $\msg{\texttt{activate!}}$ messages.
\end{lemma}
\begin{proof}
Let $Y_p$ be the number of distinct phase-$p$ trees that contain $v$. 
For $p=1$, we compute the expected number of distinct phase-$1$ trees is at most $n \gamma_1 = n \cdot \frac{\sqrt{\log n}}{n^{1-1/K}} = \sqrt{\log n} \cdot n^{1/K} = O(\sqrt{\log n})$. Using standard Chernoff bound, we can show that $Y_1$ is $O(\log n)$ with high probability.
Hence, a node can receive tree messages from at most $Y_1 =O(\log n)$ phase-$1$ trees. 

Next, we upper-bound the number of phase-$p$ trees from which $v$ receives a message:
Let $p>1$, and let $\ell$ be the number of neighbors in the $d_p$-neighborhood of $v$ that become \emph{activated} at some point in the execution. 
Moreover, we define $h_p = c\sqrt{\log n} \cdot n^{1-\frac{p-1}{K}}$, for a suitable constant $c>1$. 
If $\ell < h_p$, then the expected number of phase-$p$ trees that contains $v$ is at most $h_p \cdot \gamma_p = O(\log n \cdot n^{\frac{1}{K}}) = O(\log n)$. 
By using a standard Chernoff bound, we can show that it is also $O(\log n)$ with high probability. 
Now, we consider the case where $\ell > h_p$. 
Let $L$ be the set of the first $h_p$ \emph{activated} nodes in $v$'s $d_p$-neighborhood, ordered according to the earliest round in which they become \emph{activated}. 
We use $t^*$ to denote the latest round at which a node in $L$ starts their respective phase $(p-1)$. 
The probability that none of the nodes in $L$ is sampled to become a root in phase $p-1$ (and hence proceed to phase $p$) is at most 
\begin{align}
    \lt (1-\gamma_{p-1} \rt ) ^ {h_p} \le e^{-c\log n}=\frac{1}{n^c}.
\end{align}
Hence, with high probability, at least one node $w^*$ in $L$ is sampled as a root and initiates the construction of a phase-$(p-1)$ tree of depth $d_{p-1}$ by round $t^*$. 
This implies that at most $h_p$ nodes in $L$ can be a root of a phase-$p$ tree that contains $v$.
Since node $v$ has a distance of at most $d_p= d_{p-1}-1$, it is part of the kernel of $w^*$'s tree $T_w$, and hence becomes deactivated at the latest by round $r^* = t^*+ \Theta\lt( d_{p-1}^2 \rt)$, i.e., once the construction of $T_w$ is complete. 
Recall from Lemma~\ref{lem:deact_no_msg} that $v$ does not receive any message in any round after $r^*$. 
Thus, it remains to argue that $v$ does not receive messages from too many phase-$p$ trees prior to $r^*$.
Since any \emph{activated} node in the $d_p$-neighborhood $N_p$ of $v$ that is not in $L$ is activated no earlier than round $t^*$, it follows that any node in $N_p \setminus L$ reaches phase $p-1$ at the earliest in round $t^*$ and hence samples to become a phase $p$ root in round $r^*+1$, which, again by Lemma~\ref{lem:deact_no_msg} guarantees that $v$ can only receive a phase-$p$ tree message from trees rooted at nodes not in $L$.
In expectation, there are at most $h_p \cdot \gamma_{p} = O\lt( \log n \rt)$ such phase-$p$ trees, and the same bound holds with high probability by a Chernoff bound.
We take a union bound over all nodes, which completes the first part of the lemma.

Finally, to bound the number of $\msg{\texttt{activate!}}$ messages that $v$ may receive, we proceed analogously:
Consider of neighbors of $v$ that reach phase $K-1$, i.e., without becoming a root in an earlier phase. 
Let $N$ denote this set.
If there are at most $h_{K} = O(\log n)$ of them, then we are done, as this number is an upper bound on the number of nodes that may send an $\msg{\texttt{activate!}}$ to $v$.
Otherwise, if $h_K = \Omega\lt( \log n \rt)$, we order them in increasing order with respect to the round in which they became activated, and focus on the $h_{K}$ earliest ones.
Similarly to above, we can show by using a Chernoff bound that at least one of them, say $w$, must have become a phase-$(K-1)$ root with high probability. 
Since $d_{K-1}=3$, it follows that all other nodes in $N$ become part of the kernel of $w$'s tree and hence are deactivated before being able to send an $\msg{\texttt{activate!}}$ message.
By taking a union bound over all nodes, the lemma follows.
\end{proof}

\begin{reptheorem}{thm:pfw}
\thmPFW
\end{reptheorem}
\begin{proof} 
It is straightforward to verify that the algorithm ensures that every node wakes up eventually, since awake nodes that are unsuccessful in sampling to become a root, will simply broadcast an $\msg{\texttt{activate!}}$ message to all their neighbors.
Moreover, the time complexity bound follows directly from Lemma~\ref{lem:pfw_time}.
Lemma~\ref{lem:pfw_tree_msgs} tells us that each node $v$ receives messages from at most $O\lt( \log n \rt)$ distinct phase-$p$ trees with high probability. 
According to the way tree messages are sent, a node receives at most $O(\log n)$ messages during the construction of an individual tree, i.e., $O(1)$ messages for each tree level. 
As there are $\Theta\lt( \log n \rt)$ phases, this amounts to $O(\log^3 n)$ messages in total per node.
Moreover the number of  $\msg{\texttt{activate!}}$ messages is bounded by $O(\log n)$ per node.
It follows that the total number of messages is $O\lt( n \log^3n \rt)$.
\end{proof}

%% file: up.tex
\section{Asynchronous Algorithms with Advice in $\kt_0$ $\congest$} \label{sec:advice}

In this section, we consider advising schemes in the $\kt_0$ $\congest$ model, where an oracle first observes the entire graph and provides advice to each node before the start of the execution.  
Recall from Section~\ref{sec:intro} that under the $\kt_0$ assumption, each node is unaware of its neighbors IDs initially, and addresses its neighbors by using integer port numbers instead. 
In particular, the adversary first chooses the network and determines the port mapping at each node, which can then be used by the oracle when equipping nodes with advice. 
However, the oracle does not know the initially-awake nodes.

Throughout this section, we use $\deg_T(v)$ to denote the degree of $v$ restricted to a subgraph $T$, and use $\advice(v)$ to refer to $v$'s advice string.

\onlyLong{
\begin{reptheorem}{thm:scheme_basic} 
  \thmSchemeBasic
\end{reptheorem}
}

\subsection{Proof of Theorem~\ref{thm:scheme_broadcast}(A)}

    Let $T$ be a BFS tree rooted at an arbitrary node.
    When a node $v$ has at most $\sqrt{n}$ neighbors in $T$, we say it is a \emph{low degree tree node}, and its advice is the list of port numbers $\port_v(u)$ that lead to the neighbor $u$ of $v$ in $T$.
    If, on the other hand, $v$ has more than $\sqrt{n}$ neighbors in $T$, then we only give it a single bit set to $1$ as advice and call it a \emph{high degree tree node}.
    The bound on the maximum length of advice is immediate. 
    Moreover, the sum of the length of advice over all nodes is bounded by $\sum_v \deg_T(v) \log n$ which is $O(n\log n)$, implying the claimed bound on the average length.

Equipped with the advice, the nodes proceed as follows:
Every low degree tree node $v$ sends a message to each incident edge identified by the port numbers specified in $\advice(v)$, whereas every high degree tree node will simply broadcast over all its incident edges.

Since wake-up messages are propagated along the edges of a BFS tree, it is clear that all nodes wake up in $O(D)$ time.
It remains to show the message complexity bound:
    Let $\beta$ be the number of high degree tree nodes, and observe that $\beta = O(\sqrt{n})$, since $T$ has $n-1$ edges. 
    Therefore, the total number of messages sent is bounded by $\beta \cdot n + (n-\beta)\sqrt{n}$, implying a message complexity of $O\lt(n^{3/2}\rt).$

\subsection{Proof of Theorem~\ref{thm:scheme_cen}(B): Achieving a maximum advice length of $O(\log n)$}
The reason why we incurred a polynomial maximum advice length in Theorem~\ref{thm:scheme_broadcast} is that we somehow needed to encode the port numbers of the children in the BFS tree of a node, which may be hidden among a large number of other neighbors in the network.
To overcome this obstacle, we introduce a technique that we call \emph{child encoding}, which allows a node to recover its BFS-children at the cost of adding a logarithmic overhead in time, while requiring a maximum advice length of just $O(\log n)$ bits.

\subsubsection{Child Encoding Scheme ($\cen$)}
The advice at each node $w$ is represented as a tuple $(p_w,\fc_w,\nxt_w)$, 
where $p_w$ is the port number at $w$ that leads to the \emph{parent of $w$},
$\fc_w$ is the port number at $w$ that leads to the \emph{first child of $w$}. 
Finally, $\nxt_w$ is a pair of port numbers (not at $w$ but at its parent) that are associated with a pair of nodes that we call \emph{next siblings}.
We will formally describe the assignment of these variables in the proof of Lemma~\ref{lem:cen}. 
See Figure~\ref{fig:cen} for a more concrete example of how the oracle assigns the advice to the nodes of the first two layers of a BFS tree.

Next, we outline how the distributed algorithm will use these advice tuples: 
When a node $w$ wakes up, it first sends a wake-up message to its parent using $p_w$ and then uses $\fc_w$ to wake up one of its children $u \in T$. 
Node $u$ is also equipped with an advice tuple and, upon receiving a message from $w$, it will forward $\nxt_u$ to $w$, thus enabling $w$ to discover additional children among its neighbors.
After $w$ has received a response from $u$, it will use the port numbers in $\nxt_u$ to contact the corresponding two children nodes simultaneously, which will prompt its children to each reveal the pair of port numbers stored in their respective $\nxt$ variables to $w$ and so forth, effectively doubling the number of children that $w$ learns about in each time step. 
Note that we start executing this process simultaneously at each awake node.

\begin{lemma} \label{lem:cen}
Consider a node $v$ and an arbitrary subset $C$ of its neighbors.
There exists a procedure called \emph{child encoding}, denoted by $\cen(v,C)$, that assigns $O\lt( \log n \rt)$ bits to $v$ as well as to each node in $C$ such that, given the resulting advice,
\begin{itemize}
\item every node in $C$ knows $v$, and
\item $v$ can learn about its ports that lead to nodes in $C$ in $O(\log n)$ time by sending $O(|C|)$ messages. %
\end{itemize}
\end{lemma}
\onlyLong{
\begin{proof}
Let $u_1, u_2, \ldots, u_{|C|}$ be an arbitrary enumeration of $C$.
We first describe how to compute $\advice(v)$, which, as explained above is a tuple $(p_v,\fc_v,\nxt_v)$:
Variable $\fc_v$ contains the port number at $v$ that leads to $u_1$, i.e., $\port_v(u_1)$, which we conceptually think of as the \emph{first child of $v$}.
For now, we leave variable $p_v$ empty and initialize the tuple $\nxt_v$ to the empty list $\langle\rangle$.
Note that we may assign values to $p_v$ and $\nxt_v$ when computing $\cen(x,C_x)$, where $x$ is $v$'s parent in the tree $T$, and $C_x$ are the children of $x$, which includes $v$. 

Next, we describe how to compute the advice for $u_i \in C$.
We set $p_{u_i}$ to $\port_{u_i}(v)$, and $\fc_{u_i} = \langle\rangle$ (as mentioned, empty variables may be updated during other invocations of $\cen$). 
Finally, we assign $\nxt_{u_i}= \langle \port_{v}({u_{2i}}),\port_{v}({u_{2i+1}}) \rangle$ if $i \le \lt \lceil (|C|-1)/{2} \rt \rceil$; otherwise, we leave $\nxt_{u_i}$ empty.
Since each advice tuple consists of $O\lt( 1 \rt)$ port numbers, it follows that the maximum length of advice assigned to any node is $O\lt( \log n \rt)$ bits. 
Note that it is straightforward to separate the port numbers in the advice string by using a special delimiter symbol, which does not increase the length of advice asymptotically.

We now analyze the time complexity. 
Let $t$ be the time at which $v$ sends a message to its first child $\fc_{v}$.
Node $v$ will receive a response from $u_1$ containing the port numbers to $u_2$ and $u_3$ by time $t+2$.
Subsequently, $v$ contacts $u_2$ and $u_3$ and receives the port number to its next four children within two additional time units.
It is straightforward to show that $v$ will have discovered ports leading to (at least) $2^j$ children by time $t+2j$, for any $j \le \lceil \log_2 |C| \rceil$, and this implies that $v$ knows all ports of its children in $C$ by time $t + O(\log n)$. 

To see why the claimed bound on the message complexity holds, it suffices to notice that each edge to a child of $v$ is traversed exactly twice throughout this process.
\end{proof}
}
\begin{figure}[t]
  \centering
\includegraphics[scale=1.25]{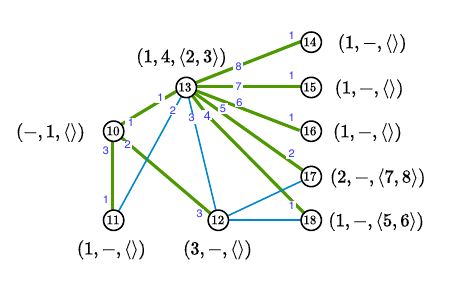}
\caption{
The child encoding scheme obtained for a BFS tree rooted at the node with ID $10$, where the green thick edges are part of the tree. The small integers are the port numbers at the respective nodes. To compute the advice, the oracle invokes $\cen(v,C)$ for each node $v$ and its set of children $C$ in the BFS tree. 
}
\label{fig:cen}
\end{figure}

We are now ready to analyze the advising scheme:
After fixing the BFS tree, the oracle computes the advice by instantiating $\cen(v,C)$, for each node $v$ and its set of children $C$ in the BFS tree.
Lemma~\ref{lem:cen} tells us that we may incur an overhead of $O(1)$ messages for each of the $n$ invocations of $\cen$, and since wake-up messages are only sent over the BFS tree edges, where each edge of $T$ is used at most twice, the message complexity is $O(n)$.

Finally, we argue the time complexity. 
Consider any sleeping node $v$ and let $P$ be the shortest path in the BFS tree $T$ to its closest initially-awake node $u$, breaking ties arbitrarily.
Upon awakening, $u$ starts to recover its incident tree edges.
According to Lemma~\ref{lem:cen}, this incurs a slowdown of at most $O(\log n)$ units of time, and the same is true for each of the nodes in $P$.
Thus, the time complexity is $O(D \log n)$, completing the proof of Theorem~\ref{thm:scheme_cen}(B).

\onlyLong{
\paragraph{Remark.} 
As explained in Section~\ref{sec:prelim}, we assume that the oracle is oblivious to the set of nodes $A_0$ that are woken up by the adversary, when computing the advice.
If the oracle \emph{does} know $A_0$, we can obtain an improved time complexity in Theorems~\ref{thm:scheme_basic}, where the diameter $D$ is replaced with the awake distance $\arad$. 
In more detail, the oracle can partition the network into a set of BFS trees $\{T_v \mid v \in A_0\}$, where $T_v$ is a BFS tree rooted at $v$, and each node $u \notin A_0$ is part of the tree for which it has the closest distance to the root, breaking ties arbitrarily.
}

%% file: spanner.tex
\subsection{A Tradeoff between Time, Messages, and Advice} \label{sec:spanner}
We now give an advising scheme that achieves an upper bound that scales with the awake distance, under the assumption that the oracle does \emph{not} know the set of initially-awake nodes.
More specifically, we give a tradeoff between time, message complexity, and advice length. 
Our result utilizes the seminal work of Baswana and Sen~\cite{DBLP:journals/rsa/BaswanaS07}, who give a randomized construction of a multiplicative $(2k-1)$-spanner. 
Below, we sketch the main ideas and introduce some notation used in the advising scheme defined in Section~\ref{sec:spanner_advising}.
We mainly follow the presentation of the Baswana-Sen algorithm given in \cite{censor2020derandomizing}, and we will also apply their elegant derandomization result.

The main idea of the Baswana-Sen algorithm is to maintain a clustering of nodes, where a \emph{cluster} is a set of nodes that form a connected subgraph, and a \emph{clustering} is a collection of disjoint clusters. 
Each cluster $C$ has a cluster leader $c$ and we maintain a spanning tree of $C$ rooted at $c$.  
The algorithm proceeds in $k$ iterations by iteratively refining the clustering. 
We use $\mathcal{C}_i$ to denote the clustering at iteration $i$ and, initially, every node forms its own cluster, i.e., $\mathcal{C}_0 = V(G)$.

The algorithm is structures into two phases, whereby Phase~1 consists of iterations $1,\dots,k-1$, while Phase~2 only comprises the $k$-th iteration, described below. 
In iteration $i \le k-1$, we obtain $\mathcal{C}_{i}$ from $\mathcal{C}_{i-1}$ by sampling each cluster with probability $n^{-1/k}$. 
We call a node that is part of a sampled cluster an \emph{$i$-clustered node} and \emph{$i$-unclustered} otherwise, following the notation in \cite{censor2020derandomizing}.
Moreover, we say that $u$ is \emph{defeated in iteration $i$}, if $u$ is $(i-1)$-clustered but $i$-unclustered.

\begin{itemize}
\item Iteration $i \in [1,k-1]$: We add edges for each defeated node $u$ by distinguishing two cases:
\begin{enumerate}
  \item[\textit{Case (A):}] If $u$ is incident to an $i$-clustered node $v$ in some cluster $C \in \mathcal{C}_{i}$, then we keep exactly one edge from $v$ to cluster $C$, by adding $\set{u,v}$ to the spanner and discarding all other edges incident to $u$ that have the other endpoint in some node in $C$.
  Conceptually, we can think of $u$ joining the cluster $C$.
  For the purpose of defining our advising scheme (described below), we treat the \emph{intra-cluster edge} $\set{u,v}$ as a directed edge $(u \to v)$ and say that $v$ is $u$'s parent in $C$.
We remark that the edge directions are only used in our advising scheme in Section~\ref{sec:spanner_advising} and do not have any relevance for the spanner itself.

\item[\textit{Case (B):}] Otherwise, if $u$ does not have any $i$-clustered neighbors, then, for each $C' \in \mathcal{C}_{i-1}$, we add one \emph{outgoing inter-cluster edge $e_{C'}$} to the spanner that connects $u$ to a neighbor $v$ who is a member of $C'$, and we call $e_{C'}$ an \emph{incoming inter-cluster edge} for $v$.
  \end{enumerate}  

\item Iteration $k$: 
  We deterministically define $\mathcal{C}_k=\emptyset$. Then, we proceed, for each node $u$ that was defeated in iteration $k$ in exactly the same way as in Case (B) for the earlier iterations described above, i.e., we keep one edge to each incident cluster in $\mathcal{C}_{k-1}$.
\end{itemize}

The randomized Baswana-Sen algorithm described above yields a spanner with $O(k\,n^{1+1/k})$ edges in expectation, which, if used directly, would provide only a bound on the \emph{expected} advice length per node.  
To obtain a deterministic upper bound instead, we leverage the derandomized version presented in \cite{censor2020derandomizing}, which provides a crucial bound on the number of outgoing inter-cluster edges that are included by defeated nodes:

\begin{lemma}[implicit in the proof of Theorem~1.6 of \cite{censor2020derandomizing}] \label{lem:derandomized_spanner}
Derandomizing the Baswana-Sen algorithm yields a $(2k-1)$-spanner algorithm with $O(k\,n^{1+1/k}\log n)$ edges, and it holds that, in each iteration, the maximum number of outgoing inter-cluster edges added to the spanner by any vertex is at most $O(n^{1/k}\log n)$. 
\end{lemma}

\subsubsection{The Advising Scheme} \label{sec:spanner_advising}
The oracle first computes a $(2k-1)$-spanner using the derandomized Baswana-Sen algorithm.
We now describe an advising scheme for encoding the different types of edges that were added to the spanner:
\begin{itemize}
\item \emph{Intra-cluster edges:} Recall that all intra-cluster edges form a directed tree rooted at the cluster leader.
To encode the intra-cluster edges between a node $v$ and the set of its children $D_v$, we compute the advice according to the child encoding scheme $\cen(v,D_v)$ described in Lemma~\ref{lem:cen}, which increases the advice length at $v$ and at each of the nodes in $D_v$ by an additive $O(\log n)$ term.

\item \emph{Outgoing and incoming inter-cluster edges:}
Consider a node $u$ defeated in iteration $i$ that does not have any incident $i$-clustered neighbors.
Suppose that it adds outgoing inter-cluster edges to nodes $v_1 \in C_1,\dots,v_\ell \in C_{\ell}$, where each $C_j \in \mathcal{C}_{i-1}$, as described in Case (B) above.
For each such node $v_j \in C_j$, let $D_j$ be the set of defeated nodes that added an edge to $v_j$ in some iteration.
In other words, $v_j$ has an incoming inter-cluster edge from each node in $D_j$.
We now instantiate the child encoding advising scheme $\cen(v_j,D_j)$ and add the required information to the advice strings of $v_j$ and the nodes in $D_j$. 
\end{itemize}

Given the above encoding rules, the algorithm for the wake-up problem is simple:
Upon awaking, each node $u$ inspects its advice and uses the decoding algorithm on the child encoding scheme described in Lemma~\ref{lem:cen} to recover its incident spanner edges. 
This ensures that every edge of the spanner $H$ is known to both endpoints. 
Finally, $u$ broadcasts a wake-up message (once) over its incident spanner edges.

\begin{reptheorem}{thm:advice_spanner}
  \thmAdviceSpanner
\end{reptheorem}
\begin{proof}
We first show the bound on the advice length:
As we are utilizing the child encoding scheme ($\cen$) for encoding intra-cluster edges and incoming inter-cluster edges, it follows from Lemma~\ref{lem:cen} that this information incurs $O(\log n)$ bits of advice per node. 
To see why $O(n^{1/k}\log^2 n)$ bits suffice for encoding the outgoing inter-cluster edges, recall that each defeated node $u$ adds at most $O(n^{1/k}\log n)$ outgoing edges to the spanner when it is defeated, according to Lemma~\ref{lem:derandomized_spanner}, and each node is defeated in exactly one iteration. 
Conceptually, we consider $u$ to be the child of each of its $O(n^{1/k}\log n)$ parents.
This means that $u$ needs to store $O\lt( \log n \rt)$ bits for each one of the $O(n^{1/k}\log n)$ instances of $\cen$ that it participates with (one for each parent), which requires $O(n^{1/k}\log^{2} n)$ bits, as claimed.

For the bound on the time, recall that each asleep node has a distance of $\arad$ from some initially awake node, which may be stretched to a distance of at most $k\cdot \arad$ in the spanner.
Moreover, each time a node is awoken, it may first spend $O(\log n)$ time to recover its incident edges from the child encoding scheme (see Lemma~\ref{lem:cen}) before it starts broadcasting the awake message over the spanner.
Combining these observations, we conclude that it takes $O\lt( k\cdot \arad\log n \rt)$ time in the worst case for any node to be woken up. 
 
The message complexity bound and the claim that the algorithm works in the $\congest$ model are both immediate, since each spanner edge is used $O(1)$ times when recovering the edges using the child-encoding scheme and also at most twice for broadcasting a wake-up message. 
\end{proof}

By instantiating Theorem~\ref{thm:advice_spanner} with $k=\Theta\lt( \log n \rt)$, we immediately obtain the following:

\begin{repcorollary}{cor:advice_spanner}
  \corAdviceSpanner
\end{repcorollary}

%% file: conclusion.tex
\section{Discussion and Open Problems} \label{sec:conclusion}

Our work raises several important questions about the wake-up problem that are yet to be resolved. 
In the context of the standard $\kt_1$ assumption (without advice), we have obtained an \emph{asynchronous} algorithm with $O(n\log n)$ time and message complexity (Theorem~\ref{thm:dfs_many}) and a \emph{synchronous} $O\lt( \arad \rt)$-time algorithm that sends $O\lt( n^{3/2}\sqrt{\log n} \rt)$ messages (Theorem~\ref{thm:kt1_fw}).

An interesting avenue for further exploration is due to the fact that our algorithms for the $\kt_1$ model may send messages of polynomial size. 
At the time of writing, there is no known algorithm that achieves wake-up in $o(m)$ message complexity {and} $\tilde O(n)$ time, while sending only messages of size $O(\log n)$. Interestingly, the question appears to be unresolved even if we assume synchronous rounds.

\begin{problem}
Can we design a (synchronous or asynchronous) $\kt_1$ $\congest$ algorithm for the wake-up problem that sends $o(m)$ messages and terminates in $\tilde O(n)$ time?
\end{problem}

When equipping nodes with advice in the $\kt_0$ model, we have seen several advising schemes that attain optimality in one of the three metrics (i.e., time, messages, or advice) while sacrificing the others, and we were also able to approach all three metrics up to polylogarithmic factors (see Cor.~\ref{cor:advice_spanner}).
However, an important open question is whether we can get optimal bounds in all three metrics simultaneously, without sacrificing logarithmic factors and match the trade-off between advice and messages stipulated in Theorem~\ref{thm:lb_kt0}: 

\begin{problem}
Is it possible to design an advising scheme that has a time complexity of $O(\arad)$, a message complexity of $\tilde O\lt( \frac{n^{2}}{2^{\beta}} \rt)$, and a maximum advice length of $O\lt( \beta \rt)$ bits in the asynchronous $\kt_0$ $\congest$ model?
\end{problem}

%% file: information.tex
\section{Tools and Definitions from Information Theory} \label{app:tools}

Here, we briefly state some key definitions and lemmas that we use in our lower bound argument in Section~\ref{sec:lb}.
We refer the reader to \cite{clover_book} for a more comprehensive introduction.

Consider jointly distributed random variables $W$, $X$, $Y$, and $Z$.
In this section, we follow the convention of using capitals for random variables and lowercase letters for values of random variables.
We denote the \emph{Shannon entropy of $X$} by $\HH\lt[ X \rt]$, which is defined as 
\begin{align}
\HH[ X ] = \sum_x \Pr[ X \!=\! x] \log_2(1 /\Pr[ X \!=\! x]). \label{eq:entropy}
\end{align}
Moreover, the \emph{conditional entropy of $X$ conditioned on $Y$} is defined as
  \begin{align}\label{eq:conditional_entropy}
    \HH[ X \mid Y ] &= \EE_Y[ \HH[ X \mid Y \!=\! y] ] 
  \end{align}

We are also interested in how much information one random variable $X$ reveals about a distinct random variable $Y$, and vice versa. 
A common way to quantify this is the \emph{mutual information} $\II[ X : Y ]$, which is defined as
  \begin{align} 
    \II[ X : Y ]
      &= \sum_{x,y}\Pr[x,y]\cdot \log\lb(\frac{\Pr[x,y]}{\Pr[x]\Pr[y]}\rb) \label{eq:mutual}
  \end{align}
Often, we want to take the expected mutual information between $X$ and $Y$ conditioned on some random variable $Z$, which is defined as the \emph{conditional mutual information of $X$ and $Y$} and, formally,
  \begin{align}
    \II[ X : Y \mid Z ]
      &= \EE_Z[ \II[ X : Y \mid Z \!=\! z ]] \label{eq:mutual_cond_exp} \\
      &= \HH[ X \mid Z ] - \HH[ X \mid Y, Z ] \label{eq:mutual_cond}.
  \end{align}

We are now ready to state some properties (without proof) that we use in Section~\ref{sec:lb}:
\begin{lemma}[see, e.g., \cite{clover_book}] \label{lem:inf_props}
Let $X, Y, W, Z$ be random variables. 
The following properties hold:
\begin{enumerate}
\item[(a)] Every encoding of $X$ has expected length at least $\HH[X]$, i.e., $\HH\lt[ X \rt] \le |X|$.
\item[(b)] $\II[ X : Y \mid Z ] \le \HH[ X \mid Z ] \le \HH[ X ]$.
\item[(c)] $\II\lt[ X : Y, Z \mid W \rt] = \II\lt[ X : Y \mid W \rt] + \II\lt[ X : Z \mid W,Y \rt]$.
\item[(d)] $\II\lt[ X : Y, Z \mid W \rt] \ge \II\lt[ X : Y \mid W \rt]$.
\item[(e)] $\HH\lt[ X,Y \rt] = \HH\lt[ X \rt] + \HH\lt[ Y \mid X \rt] \le \HH\lt[ X \rt] + \HH\lt[ Y \rt].$
\item[(f)] $\HH\lt[ X\rt] \le \log_2\lt(\supp(X)\rt),$ where $\supp(X)$ denotes the size of the support of $X$. Equality holds if $X$ is uniformly distributed.
\item[(g)] If conditioned on $Z$, $Y$ is independent of $W$, then  $\II\lt[ X : Y \ \md|\ W, Z \rt] \ge \II\lt[ X : Y \ \md|\ Z \rt]$.  
\end{enumerate}
\end{lemma}

%% file: basic.tex
\section{Proof of Corollary~\ref{cor:basic}} \label{app:basic}
We give a simple advising scheme very similar to the one proposed in \cite{fraigniaud2006oracle} that enables solving the wake-up problem with both, optimal time and message complexity. 
While the maximum advice length may be up to $O(n)$ bits, we show that it is only $O(\log n)$ bits on average.

    Consider the following advising scheme where the oracle constructs a BFS tree $T$ with an arbitrary node as the root.  
    The advice $\advice(v)$ of a node $v$ is defined in such a way that it can identify $v$' neighbors in $T$ using $O(n)$ bits. 
    When a node has less than or equal to $\frac{n}{\log n}$ neighbors in $T$, we say it is a \emph{low degree tree node}, otherwise, it is a \emph{high degree tree node}. 
    Note that for a node to know if it is a low degree or high degree tree node, it has to know the exact value of $n$. 
    This can be achieved by attaching the bit string representation of $n$ (which is of size $O(\log n)$) to each node's advice. 
    If $v$ is a low degree tree node, then we append the list of port numbers at $v$ that lead to its neighbors in $T$ to $\advice(v)$ . 
    On the other hand, if $v$ is a high degree tree node, let $u_1, \ldots, u_d$ be $v$'s neighbors in $G$, such that $\port_v(u_i)=i$. We append a $n$-bit string $B$ to $\advice(v)$, where $B$ is such that the $i$-th bit is $1$ if $u_i$ is $v$'s neighbor in $T$, $0$ otherwise.
    The distributed algorithm that solves the adversarial wake-up problem is such that when a node is awake, it will wake up all the neighbors given in its advice. 
    
    It is straightforward to verify that all nodes are awake in $O(D)$ time since wake-up messages are propagated along the edges of a BFS tree.
    Regarding the number of messages, recall that $T$ has $n-1$ edges.
    Since only the edges in $T$ are used to transmit messages, and each of them is used exactly twice throughout the entire execution, we conclude that the message complexity is $O(n)$. 

    Finally, we show the claimed bound on the advice length: 
    The maximum length of advice of each node is $O(n)$ by constructions. 
    Let $\beta$ be the number of high degree tree nodes. We have $\beta = O(\log n)$ since $T$ has $n-1$ edges.
    Hence, the sum of the length of advice of all the nodes is $O(n \log n)$, which gives $O(\log n)$ average length of advice per node.

%% file: goodkt0.tex
\section{Proof of Lemma~\ref{lem:good_kt0}} \label{app:good_kt0}
\begin{replemma}{lem:good_kt0}
  \lemGoodKTZero
\end{replemma}
\begin{proof} 
Assume towards a contradiction that the statement is false. Let $\bar{S} = V \setminus S$ be the set of nodes such that $\Pr\lt[ \good_i \rt]< \frac{1}{2\log_2 n}$, and note that $|\bar{S}|> \frac{1}{2} n$. 
We define $M_{v}$ to be the number of messages sent and received by node $v$, and also define $M=\sum_{v \in U \cup V \cup W} M_v$. 
Furthermore, let $\mathbf{P}$ be the random variable that represents the port assignments of all the nodes.
Thus, we know that $M_{v}$ is a deterministic function of $B$, $\mathbf{P}$, and $\mathbf{Y}$, and note that $\mathbf{Y}$ is itself a function of $\mathbf{P}$, since the graph topology and the node IDs are fixed in our lower bound graph.
This means that $M_v$ only depends on $B$ and $\mathbf{P}$.

We have
\begin{align}
\EE_{\mathbf{P},B}\lt[ M \rt] 
&= 
	\sum_{p}\Pr\lt[ \mathbf{P} \!=\! p \rt] 
  \sum_{{b}} \Pr\lt[ B \!=\! b \ \md|\ p \rt] 
  \cdot M(p,b)\notag\\ 
\ann{since $B \perp \mathbf{P}$}
&= 
	\sum_{p}\Pr\lt[ \mathbf{P} \!=\! p \rt] 
  \sum_{b} \Pr\lt[ B \!=\! b \rt] 
  \cdot M(p,b)\notag\\ 
&= 
	\sum_{p}\Pr\lt[ \mathbf{P} \!=\! p \rt] 
  \EE_{B}\lt[ M(p,B) \rt] \notag\\ 
&\le \frac{n^2}{2^{\beta+3}\log_2 n} + O(n), \label{eq:exp_up}
\end{align}
where the final inequality follows since $\EE_{B}\lt[ M(p,B) \rt]$ is precisely twice the expected message complexity of the algorithm for the given port assignment $p$, which is at most $\frac{n^2}{2^{\beta+4}\log_2n} + O(n)$ by assumption.

On the other hand, we also have
\begin{align}
\EE_{\mathbf{P},B}\lt[ M \rt] 
&\ge
	\sum_{v_i \in \bar{S}} \EE_{\mathbf{P},B}\lt[ M_{v_i} \rt] \notag\\ 
&\ge
	\sum_{v_i \in \bar{S}} \sum_{p,b} 
  \Pr\lt[  p,b \rt]   
  M_{v_i}(p,b)  \notag\\ 
&\ge
	\sum_{v_i \in \bar{S}}\bigg(
  \sum_{\substack{p,b\colon\\ M_{v_i}\!(p,b)> \frac{n }{2^{\beta}}}} \!\!
  \Pr\lt[  p,b \rt]   
  M_{v_i}(p,b) \bigg) \notag\\ 
&>
\frac{n }{2^{\beta}} \cdot
	\sum_{v_i \in \bar{S}}\bigg(  
  \sum_{\substack{p,b\colon\\ M_{v_i}\!(p,b)> \frac{n }{2^{\beta}}}} \!\!
  \Pr\lt[  p,b \rt]   \bigg) \notag\\ 
&=
  \frac{n^2 }{2^{\beta+2}\log_2n} , \label{eq:exp_lb}
\end{align}
where the final step follows from the assumption that $|\bar{S}|\ge \frac{n}{2}$ and because 
\[
\Pr[  M_{v_i}\!(p,b)> \frac{n }{2^{\beta}}] = \Pr\lt[ \neg\good_i \rt] \ge \frac{1}{2\log_2n},
\]
  for every $v_i \in \bar{S}$. 
Combining \eqref{eq:exp_up} and \eqref{eq:exp_lb} yields a contradiction.
\end{proof}

%% file: lb_algo.tex
\section{Solving Wake-up on $\mathcal{G}_k$ with Polylogarithmic Overhead in the $\local$ Model} \label{app:lb_algo}
First, each node in $V$ sends a message to $\Theta( \log^2 n)$ randomly chosen neighbors. 
Since the subgraph $G[U \cup V]$ is a regular bipartite graph, a straightforward balls-into-bins argument shows that every node in $U$ wakes up with high probability.
Then, all awake nodes execute the gossip algorithm of \cite{haeupler2015simple} that solves the 1-local broadcast problem, which means that, initially, every node holds a unique rumor and, eventually, every node knows the rumors of all its neighbors.
We add the modification to this algorithm that every node in $V$ chooses two (instead of just one) neighbors to communicate with in each iteration.
Since this ensures that every node $v \in V$ always chooses at least one edge in the cut $(U,V)$, node $v$ learns at least as many rumors in a given round, as if we were executing the (unmodified) gossip algorithm of \cite{haeupler2015simple} on the subgraph $G[U \cup V]$ (i.e., without the edges to the nodes in $W$).
Consequently, the nodes in $V$ learn about all of their respective neighbors in $U$ in $\Theta\lt( \log^2n \rt)$ rounds and by sending $\Theta\lt( n\log^2n \rt)$ messages. This, in turn, ensures that every node in $V$ can identify its crucial neighbor in $W$ in one additional round.

%% file: app_kt1_algo.tex
\section{Proof of Claim~\ref{cl:forward}} \label{app:kt1_algo}

To prove the claim, consider some node $v$ and let $X$ denote the number of tokens that node $v$ forwards in the execution. 
We first argue that $\EE\lt [ X \rt] = O(\log n)$.
Consider the list of distinct tokens that visited node $v$ throughout the execution, sorted in the order of their first arrivals.
Let $\rho_{i}$ be the rank carried by the $i$-th token and $X_i$ be the indicator random variable $v$ forwards this token. 
According to our algorithm, $v$ forwards this token only if $\rho_i > \rho_j$ for all $j<i$. 
By Claim~\ref{cl:rank}, we have $\Pr\lt[ X_i = 1 \rt] = \frac{1}{i}$. 
Hence, $\EE\lt [ X \rt] = \EE \lt [\sum_{i=1}^{d} X_i\rt] = \sum_{i=1}^{d} \frac{1}{i} = O(\log d)$, where $d$ is the number of tokens that visit node $v$. 
Now, Claim~\ref{cl:forward} follows by a standard Chernoff bound~\cite{mitzenmacher} since $X_1, X_2, \ldots, X_d$ are independent indicator random variables and $d\le n$.

%% file: app_kt1_sync.tex
\section{Omitted Proofs from Section~\ref{sec:kt1_algo_sync}} \label{app:kt1_algo_sync}

\begin{replemma}{lem:deactivated}
\lemDeactivated
\end{replemma}
\begin{proof}
According to the description of the algorithm, there are only two ways $u$ can become deactivated in round $r$:
\begin{enumerate} 
\item Node $u$ has joined a BFS tree $T$ as a root, level-$1$ or level-$2$ node in some earlier round, and the construction of $T$ completed in round $r$. 
\item $u$ has broadcast an $\msg{\texttt{activate!}}$ message to its neighbors in round $r-1$. 
\end{enumerate}
In either case, all of $u$'s neighbors are guaranteed to be awake in round $r$. 
\end{proof}

\begin{replemma}{lem:9round}
\lemRound
\end{replemma}
\begin{proof}
It takes $1$ round for the root to send a message to all of its neighbors and an additional $2$ rounds of communication for the level-$1$ nodes to be informed of the BFS edges $S_{2}$ which are the BFS edges between the level-$1$ nodes and the level-$2$ nodes. Next, it takes $1$ round for the level-$1$ nodes to send messages across all incident edges of $S_{2}$ to the level-$2$ nodes, and subsequent  $4$ rounds of communication over the BFS edges between the root and the level-$2$ nodes for the level-$2$ nodes to be informed of the BFS edges $S_{3}$.
Finally, it takes $1$ round for the level-$2$ nodes to send messages across all incident edges of $S_{3}$ to level-$3$ nodes. 
This sums up to $9$ rounds in total. 
\end{proof}

\begin{replemma}{lem:active}
\lemActive
\end{replemma}
\begin{proof}
When a node $u$ wakes up in round $r$, its status is either activated or deactivated. 
If $u$ is activated and is sampled as a root, it initiates the BFS tree construction and, by Lemma~\ref{lem:9round}, becomes deactivated at the end of round $r+9$.
If $u$ is not sampled as a root and remains activated for $9$ rounds, then it executes the broadcast step, and thus $u$ becomes deactivated at the end of round $r+10$. 
\end{proof}

%% file: main.bbl
\newcommand{\etalchar}[1]{$^{#1}$}
\begin{thebibliography}{DKMJ{\etalchar{+}}22}

\bibitem[{Adv}95]{amd_20213}
{Advanced Micro Devices, Inc.\ (AMD)}.
\newblock Magic packet technology.
\newblock Technical report, 1995.
\newblock
  \url{https://www.amd.com/content/dam/amd/en/documents/archived-tech-docs/white-papers/20213.pdf}.
  Accessed: 2024-10-11.

\bibitem[AG91]{afek1991time}
Yehuda Afek and Eli Gafni.
\newblock Time and message bounds for election in synchronous and asynchronous
  complete networks.
\newblock {\em SIAM Journal on Computing}, 20(2):376--394, 1991.

\bibitem[AGM12]{AGM-soda12}
Kook~Jin Ahn, Sudipto Guha, and Andrew McGregor.
\newblock Analyzing graph structure via linear measurements.
\newblock In {\em Proceedings of the twenty-third annual ACM-SIAM symposium on
  Discrete Algorithms}, pages 459--467. SIAM, 2012.

\bibitem[AGPV90]{AGPV88}
Baruch Awerbuch, Oded Goldreich, David Peleg, and Ronen Vainish.
\newblock A trade-off between information and communication in broadcast
  protocols.
\newblock {\em J. {ACM}}, 37(2):238--256, 1990.

\bibitem[AHK{\etalchar{+}}20]{DBLP:conf/soda/AugustineHKSS20}
John Augustine, Kristian Hinnenthal, Fabian Kuhn, Christian Scheideler, and
  Philipp Schneider.
\newblock Shortest paths in a hybrid network model.
\newblock In Shuchi Chawla, editor, {\em Proceedings of the 2020 {ACM-SIAM}
  Symposium on Discrete Algorithms, {SODA} 2020, Salt Lake City, UT, USA,
  January 5-8, 2020}, pages 1280--1299. {SIAM}, 2020.

\bibitem[BBK{\etalchar{+}}24]{10.1145/3662158.3662805}
Alkida Balliu, Sebastian Brandt, Fabian Kuhn, Krzysztof Nowicki, Dennis
  Olivetti, Eva Rotenberg, and Jukka Suomela.
\newblock Brief announcement: Local advice and local decompression.
\newblock In {\em Proceedings of the 43rd ACM Symposium on Principles of
  Distributed Computing}, PODC '24, page 117–120. Association for Computing
  Machinery, 2024.

\bibitem[BS07]{DBLP:journals/rsa/BaswanaS07}
Surender Baswana and Sandeep Sen.
\newblock A simple and linear time randomized algorithm for computing sparse
  spanners in weighted graphs.
\newblock {\em Random Struct. Algorithms}, 30(4):532--563, 2007.

\bibitem[CDH{\etalchar{+}}18]{DBLP:conf/podc/ChangDHHLP18}
Yi{-}Jun Chang, Varsha Dani, Thomas~P. Hayes, Qizheng He, Wenzheng Li, and Seth
  Pettie.
\newblock The energy complexity of broadcast.
\newblock In Calvin Newport and Idit Keidar, editors, {\em Proceedings of the
  2018 {ACM} Symposium on Principles of Distributed Computing, {PODC} 2018,
  Egham, United Kingdom, July 23-27, 2018}, pages 95--104. {ACM}, 2018.

\bibitem[CGP20]{DBLP:conf/podc/ChatterjeeGP20}
Soumyottam Chatterjee, Robert Gmyr, and Gopal Pandurangan.
\newblock Sleeping is efficient: {MIS} in \emph{O}(1)-rounds node-averaged
  awake complexity.
\newblock In Yuval Emek and Christian Cachin, editors, {\em {PODC} '20: {ACM}
  Symposium on Principles of Distributed Computing, Virtual Event, Italy,
  August 3-7, 2020}, pages 99--108. {ACM}, 2020.

\bibitem[CHKM12]{DBLP:conf/stoc/Censor-HillelHKM12}
Keren Censor{-}Hillel, Bernhard Haeupler, Jonathan~A. Kelner, and Petar
  Maymounkov.
\newblock Global computation in a poorly connected world: fast rumor spreading
  with no dependence on conductance.
\newblock In Howard~J. Karloff and Toniann Pitassi, editors, {\em Proceedings
  of the 44th Symposium on Theory of Computing Conference, {STOC} 2012, New
  York, NY, USA, May 19 - 22, 2012}, pages 961--970. {ACM}, 2012.

\bibitem[CHPS20]{censor2020derandomizing}
Keren Censor-Hillel, Merav Parter, and Gregory Schwartzman.
\newblock Derandomizing local distributed algorithms under bandwidth
  restrictions.
\newblock {\em Distributed Computing}, 33(3):349--366, 2020.

\bibitem[Coh97]{cohen1997size}
Edith Cohen.
\newblock Size-estimation framework with applications to transitive closure and
  reachability.
\newblock {\em Journal of Computer and System Sciences}, 55(3):441--453, 1997.

\bibitem[CT06]{clover_book}
T.~Cover and J.A. Thomas.
\newblock {\em Elements of Information Theory, second edition}.
\newblock Wiley, 2006.

\bibitem[DKMJ{\etalchar{+}}22]{dufoulon2022almost}
Fabien Dufoulon, Shay Kutten, William~K Moses~Jr, Gopal Pandurangan, and David
  Peleg.
\newblock An almost singularly optimal asynchronous distributed mst algorithm.
\newblock In {\em 36th International Symposium on Distributed Computing}, 2022.

\bibitem[DPP{\etalchar{+}}24]{DBLP:conf/innovations/DufoulonPPP024}
Fabien Dufoulon, Shreyas Pai, Gopal Pandurangan, Sriram~V. Pemmaraju, and Peter
  Robinson.
\newblock The message complexity of distributed graph optimization.
\newblock In Venkatesan Guruswami, editor, {\em 15th Innovations in Theoretical
  Computer Science Conference, {ITCS} 2024, January 30 to February 2, 2024,
  Berkeley, CA, {USA}}, volume 287 of {\em LIPIcs}, pages 41:1--41:26. Schloss
  Dagstuhl - Leibniz-Zentrum f{\"{u}}r Informatik, 2024.

\bibitem[DPRS24]{opodis2024}
Fabien Dufoulon, Gopal Pandurangan, Peter Robinson, and Michele Scquizzato.
\newblock The singular optimality of distributed computation in {LOCAL}.
\newblock In {\em Proceedings of the 28th International Conference on
  Principles of Distributed Systems (OPODIS)}. LIPIcs, (to appear), 2024.

\bibitem[Elk06]{elkin2006faster}
Michael Elkin.
\newblock A faster distributed protocol for constructing a minimum spanning
  tree.
\newblock {\em Journal of Computer and System Sciences}, 72(8):1282--1308,
  2006.

\bibitem[FGIP09]{fraigniaud2009distributed}
Pierre Fraigniaud, Cyril Gavoille, David Ilcinkas, and Andrzej Pelc.
\newblock Distributed computing with advice: information sensitivity of graph
  coloring.
\newblock {\em Distributed Computing}, 21:395--403, 2009.

\bibitem[FIP06]{fraigniaud2006oracle}
Pierre Fraigniaud, David Ilcinkas, and Andrzej Pelc.
\newblock Oracle size: a new measure of difficulty for communication tasks.
\newblock In {\em Proceedings of the twenty-fifth annual ACM symposium on
  Principles of distributed computing}, pages 179--187, 2006.

\bibitem[GHBK12]{gandhi2012sleep}
Anshul Gandhi, Mor Harchol-Balter, and Michael~A Kozuch.
\newblock Are sleep states effective in data centers?
\newblock In {\em 2012 international green computing conference (IGCC)}, pages
  1--10. IEEE, 2012.

\bibitem[GHS83]{gallager1983distributed}
Robert~G. Gallager, Pierre~A. Humblet, and Philip~M. Spira.
\newblock A distributed algorithm for minimum-weight spanning trees.
\newblock {\em ACM Transactions on Programming Languages and systems (TOPLAS)},
  5(1):66--77, 1983.

\bibitem[Hae15]{haeupler2015simple}
Bernhard Haeupler.
\newblock Simple, fast and deterministic gossip and rumor spreading.
\newblock {\em Journal of the ACM (JACM)}, 62(6):1--18, 2015.

\bibitem[HPP{\etalchar{+}}15]{podc15}
James~W Hegeman, Gopal Pandurangan, Sriram~V Pemmaraju, Vivek~B Sardeshmukh,
  and Michele Scquizzato.
\newblock Toward optimal bounds in the congested clique: Graph connectivity and
  mst.
\newblock In {\em Proceedings of the 2015 ACM Symposium on Principles of
  Distributed Computing}, pages 91--100, 2015.

\bibitem[KKM{\etalchar{+}}12]{DBLP:journals/dc/KhanKMPT12}
Maleq Khan, Fabian Kuhn, Dahlia Malkhi, Gopal Pandurangan, and Kunal Talwar.
\newblock Efficient distributed approximation algorithms via probabilistic tree
  embeddings.
\newblock {\em Distributed Comput.}, 25(3):189--205, 2012.

\bibitem[KKM13]{KKM-soda13}
Bruce~M Kapron, Valerie King, and Ben Mountjoy.
\newblock Dynamic graph connectivity in polylogarithmic worst case time.
\newblock In {\em Proceedings of the twenty-fourth annual ACM-SIAM symposium on
  Discrete algorithms}, pages 1131--1142. SIAM, 2013.

\bibitem[KKT15]{DBLP:conf/podc/KingKT15}
Valerie King, Shay Kutten, and Mikkel Thorup.
\newblock Construction and impromptu repair of an {MST} in a distributed
  network with o(m) communication.
\newblock In Chryssis Georgiou and Paul~G. Spirakis, editors, {\em Proceedings
  of the 2015 {ACM} Symposium on Principles of Distributed Computing, {PODC}
  2015, Donostia-San Sebasti{\'{a}}n, Spain, July 21 - 23, 2015}, pages 71--80.
  {ACM}, 2015.

\bibitem[KMJPP20]{kutten2020singularly}
Shay Kutten, William~K Moses~Jr, Gopal Pandurangan, and David Peleg.
\newblock Singularly optimal randomized leader election.
\newblock In {\em 34th International Symposium on Distributed Computing (DISC
  2020)}. Schloss-Dagstuhl-Leibniz Zentrum f{\"u}r Informatik, 2020.

\bibitem[KPP{\etalchar{+}}15]{jacm15}
Shay Kutten, Gopal Pandurangan, David Peleg, Peter Robinson, and Amitabh
  Trehan.
\newblock On the complexity of universal leader election.
\newblock {\em J. {ACM}}, 62(1):7:1--7:27, 2015.

\bibitem[KRT24]{kutten2024tight}
Shay Kutten, Peter Robinson, and Ming~Ming Tan.
\newblock Tight bounds on the message complexity of distributed tree
  verification.
\newblock {\em arXiv preprint arXiv:2401.11991}, 2024.

\bibitem[KRTZ23]{DBLP:conf/podc/Kutten0T023}
Shay Kutten, Peter Robinson, Ming~Ming Tan, and Xianbin Zhu.
\newblock Improved tradeoffs for leader election.
\newblock In Rotem Oshman, Alexandre Nolin, Magn{\'{u}}s~M. Halld{\'{o}}rsson,
  and Alkida Balliu, editors, {\em Proceedings of the 2023 {ACM} Symposium on
  Principles of Distributed Computing, {PODC} 2023, Orlando, FL, USA, June
  19-23, 2023}, pages 355--365. {ACM}, 2023.

\bibitem[KSSV00]{karp2000randomized}
Richard Karp, Christian Schindelhauer, Scott Shenker, and Berthold Vocking.
\newblock Randomized rumor spreading.
\newblock In {\em Proceedings 41st Annual Symposium on Foundations of Computer
  Science}, pages 565--574. IEEE, 2000.

\bibitem[LUW95]{girth}
Felix Lazebnik, Vasiliy~A Ustimenko, and Andrew~J Woldar.
\newblock A new series of dense graphs of high girth.
\newblock {\em Bulletin of the American mathematical society}, 32(1):73--79,
  1995.

\bibitem[MK19]{mashreghi2019brief}
Ali Mashreghi and Valerie King.
\newblock Brief announcement: Faster asynchronous mst and low diameter tree
  construction with sublinear communication.
\newblock In {\em 33rd International Symposium on Distributed Computing (DISC
  2019)}. Schloss-Dagstuhl-Leibniz Zentrum f{\"u}r Informatik, 2019.

\bibitem[MK21]{DBLP:journals/dc/MashreghiK21}
Ali Mashreghi and Valerie King.
\newblock Broadcast and minimum spanning tree with o(m) messages in the
  asynchronous {CONGEST} model.
\newblock {\em Distributed Comput.}, 34(4):283--299, 2021.

\bibitem[MU04]{mitzenmacher}
M.~Mitzenmacher and E.~Upfal.
\newblock {\em Probability and Computing: Randomized Algorithms and
  Probabilistic Analysis}.
\newblock Cambridge University Press, 2004.

\bibitem[Pel00]{peleg}
David Peleg.
\newblock {\em Distributed Computing: A Locality-Sensitive Approach}.
\newblock SIAM, Philadelphia, 2000.

\bibitem[PPPR21]{pai2021can}
Shreyas Pai, Gopal Pandurangan, Sriram~V Pemmaraju, and Peter Robinson.
\newblock Can we break symmetry with o (m) communication?
\newblock In {\em Proceedings of the 2021 ACM Symposium on Principles of
  Distributed Computing}, pages 247--257, 2021.

\bibitem[Rob21]{soda21}
Peter Robinson.
\newblock Being fast means being chatty: The local information cost of graph
  spanners.
\newblock In D{\'{a}}niel Marx, editor, {\em Proceedings of the 2021 {ACM-SIAM}
  Symposium on Discrete Algorithms, {SODA} 2021, Virtual Conference, January 10
  - 13, 2021}, pages 2105--2120. {SIAM}, 2021.

\bibitem[SHK{\etalchar{+}}12]{DBLP:journals/siamcomp/SarmaHKKNPPW12}
Atish~Das Sarma, Stephan Holzer, Liah Kor, Amos Korman, Danupon Nanongkai,
  Gopal Pandurangan, David Peleg, and Roger Wattenhofer.
\newblock Distributed verification and hardness of distributed approximation.
\newblock {\em {SIAM} J. Comput.}, 41(5):1235--1265, 2012.

\bibitem[Sin92]{singh1992leader}
Gurdip Singh.
\newblock Leader election in complete networks.
\newblock In {\em Proceedings of the eleventh annual ACM symposium on
  Principles of distributed computing}, pages 179--190, 1992.

\bibitem[SS11]{sauerwald2011rumor}
Thomas Sauerwald and Alexandre Stauffer.
\newblock Rumor spreading and vertex expansion on regular graphs.
\newblock In {\em Proceedings of the twenty-second annual ACM-SIAM symposium on
  Discrete Algorithms}, pages 462--475. SIAM, 2011.

\bibitem[wik24]{wikipedia_wake-on-lan}
{Wake-on-LAN} --- wikipedia{,} the free encyclopedia, 2024.
\newblock \url{https://en.wikipedia.org/wiki/Wake-on-LAN}. Accessed:
  2024-10-11.

\bibitem[Yao77]{yao1977probabilistic}
Andrew Chi-Chin Yao.
\newblock Probabilistic computations: Toward a unified measure of complexity.
\newblock In {\em 18th Annual Symposium on Foundations of Computer Science
  (sfcs 1977)}, pages 222--227. IEEE Computer Society, 1977.

\end{thebibliography}
